\documentclass[10pt,journal,compsoc]{IEEEtran}
\ifCLASSOPTIONcompsoc
  \usepackage[nocompress]{cite}
\else
  \usepackage{cite}
\fi
\ifCLASSINFOpdf
   \usepackage[pdftex]{graphicx}
\else
\fi

\usepackage{glossaries}
\usepackage{acronym}
\usepackage{amsfonts}
\usepackage{graphics}
\usepackage{amsthm}
\usepackage{amsmath}
\usepackage{pgfplots}
\usepackage{graphics}
\usepackage{graphicx}
\usepackage{float}
\usepackage{bm}
\usepackage{bbm}
\usepackage{epstopdf}
\usepackage{amssymb}
\usepackage{epsfig}
\usepackage{fancyhdr}
\usepackage{graphics}
\usepackage{helvet}
\usepackage{hhline}
\usepackage{multicol}
\usepackage{mdframed}
\usepackage{multirow}
\usepackage{mathtools}
\usepackage{tikz}
\usepackage{xfrac}
\usepackage{algorithm}
\usepackage{algorithmicx}
\usepackage{algorithm,algpseudocode}

\newacronym{5G}{5G}{Fifth-Generation}
\newacronym{4G}{4G}{Fourth-Generation}
\newacronym{eMBB}{eMBB}{Enhanced Mobile BroadBand}
\newacronym{mMTC}{mMTC}{Massive Machine Type Communications}
\newacronym{URLLC}{URLLC}{Ultra-Reliable Low Latency Communications}
\newacronym{NR}{NR}{New Radio}
\newacronym{MIMO}{MIMO}{Multiple Input Multiple Output}
\newacronym{LTE}{LTE}{Long Term Evolution}
\newacronym{D2D}{D2D}{Device-to-Device}
\newacronym{M2M}{M2M}{Machine-to-Machine}
\newacronym{V2V}{V2V}{Vehicle-to-Vehicle}
\newacronym{V2X}{V2X}{Vehicle-to-Everything}
\newacronym{ProSe}{ProSe}{Proximity Services}
\newacronym{PS}{PS}{Public Safety}
\newacronym{IoT}{IoT}{Internet of Things}
\newacronym{CSI}{CSI}{Channel State Information}
\newacronym{HARQ}{HARQ}{Hybrid Automatic Repeat Request}
\newacronym{AMC}{AMC}{Adaptive Modulation Coding}
\newacronym{TDD}{TDD}{Time Division Duplex}
\newacronym{TDMA}{TDMA}{Time Division Multiple Access}
\newacronym{QoS}{QoS}{Quality of Service}
\newacronym{FDD}{FDD}{Frequency Division Duplex}
\newacronym{DTMC}{DTMC}{Discrete Time Markov Chain}
\newacronym{DRX}{DRX}{Discontinuous Reception}
\newacronym{DL}{DL}{Downlink}
\newacronym{UL}{UL}{Uplink}
\newacronym{3GPP}{3GPP}{3rd Generation Partnership Project}
\newacronym{SNR}{SNR}{Signal-to-Noise Ratio}
\newacronym{SINR}{SINR}{Signal-to-Interference-plus-Noise Ratio}
\newacronym{OP}{OP}{Optimization Problem}
\newacronym{RB}{RB}{Resource Block}
\newacronym{RE}{RE}{Resource Emplacement}
\newacronym{PUCCH}{PUCCH}{Physical Uplink Control Channel}
\newacronym{BS}{BS}{Base Station}
\newacronym{UB}{UB}{Upper Bound}
\newacronym{UE}{UE}{User Equipment}
\newacronym{RI}{RI}{Rank Index}
\newacronym{PMI}{PMI}{Precoding Matrix Indicator}
\newacronym{CQI}{CQI}{Channel Quality Index}
\newacronym{DCI}{DCI}{Downlink Control Information}
\newacronym{AWGN}{AWGN}{Additive White Gaussian Noise}
\newacronym{QPSK}{QPSK}{Quadrature Phase-Shift Keying}
\newacronym{PDCCH}{PDCCH}{Physical Downlink Control Channel}
\newacronym{PUSCH}{PUSCH}{Physical Uplink Shared Channel} 
\newacronym{RRC}{RRC}{Radio Resource Control} 
\newacronym{E-UTRA}{E-UTRA}{Evolved Universal Mobile Telecommunications System Terrestrial Radio Access,}  
\newacronym{MANET}{MANET}{Mobile Ad-Hoc Networks}  
\newacronym{OFDMA}{OFDMA}{Orthogonal Frequency Division Multiple Access}  
\newacronym{SC-FDMA}{SC-FDMA}{Single Carrier Frequency Division Multiple Access}  
\newacronym{EC}{EC}{Energy Consumption}  
\newacronym{EE}{EE}{Energy Efficiency} 
\newacronym{MU}{MU}{Master User Equipment}
\newacronym{TTI}{TTI}{Time Transmission Interval}
\newacronym{ISD}{ISD}{Inter-Site Distance}
\newacronym{RF}{RF}{Radio Frequency}
\newacronym{UAV}{UAV}{unmanned aerial vehicle}
\newacronym{CDF}{CDF}{Cumulative Distribution Function}

\newtheorem{proposition}{Proposition}
\newtheorem{theorem}{Theorem}[section]

\newtheorem{remark}{Remark}[section]

\DeclareGraphicsExtensions{.pdf,.eps}
\graphicspath{{.//}}

\algnewcommand\algorithmicswitch{\textbf{switch}}
\algnewcommand\algorithmiccase{\textbf{case}}
\algnewcommand\algorithmicassert{\texttt{assert}}
\algnewcommand\Assert[1]{\State \algorithmicassert(#1)}%
\algdef{SE}[SWITCH]{Switch}{EndSwitch}[1]{\algorithmicswitch\ #1\ \algorithmicdo}{\algorithmicend\ \algorithmicswitch}%
\algdef{SE}[CASE]{Case}{EndCase}[1]{\algorithmiccase\ #1}{\algorithmicend\ \algorithmiccase}%
\algtext*{EndSwitch}%
\algtext*{EndCase}%
\newcommand{\argmin}{\mathop{\mathrm{argmin}}}

\begin{document}
%

\title{Distributed vs. Centralized Scheduling in \gls{D2D}-enabled Cellular Networks}
%
%
%

\author{Rita Ibrahim,~\IEEEmembership{Member,~IEEE,}
        Mohamad Assaad,~\IEEEmembership{Senior Member,~IEEE,}
        Berna Sayrac,~\IEEEmembership{Member,~IEEE,}
        and~Azeddine Gati,~\IEEEmembership{Member,~IEEE}}
\maketitle

\begin{abstract}
 Employing channel adaptive resource allocation can yield to a large enhancement in almost any performance metric of \acrfull{D2D} communications. We observe that \gls{D2D} users are able to estimate their local \gls{CSI}, however the base station needs some signaling exchange to acquire this information. Based on the \gls{D2D} users' knowledge of their local \gls{CSI}, we provide a scheduling framework that shows how distributed approach outperforms centralized one. We start by proposing a centralized scheduling that requires the knowledge of \gls{D2D} links' \gls{CSI} at the base station level. This \gls{CSI} reporting suffers from the limited number of resources available for feedback transmission. Therefore, we benefit from the users' knowledge of their local \gls{CSI} to develop a distributed algorithm for \gls{D2D} resource allocation. In distributed approach, collisions may occur between the different \gls{CSI} reporting; thus a collision reduction algorithm is proposed. We give a description on how both centralized and distributed algorithms can be implemented in practice. Furthermore, numerical results are presented to corroborate our claims and demonstrate the gain that the proposed scheduling algorithms bring to cellular networks.
\end{abstract}

\begin{IEEEkeywords}
\gls{D2D} communications, resource allocation, distributed scheduling, \gls{CSI} feedback
\end{IEEEkeywords}

\IEEEdisplaynontitleabstractindextext

%
\IEEEpeerreviewmaketitle


\section{Introduction}
%
%
%
%

 


\IEEEPARstart{I}{n} cellular networks, the knowledge of the channel condition at the transmitter can improve the performance of cellular communications by allowing the transmitter to dynamically adapt its transmission scheme and providing by that a better throughput (i.e. \gls{AMC} scheme). For \gls{D2D} communications, each \gls{D2D} user has the knowledge of its local \gls{D2D} channel state. For both \gls{TDD} and \gls{FDD} cellular networks, \textit{feedback} is one way for keeping the \gls{BS} updated with the \gls{D2D} channels' measurements. Mobile users estimate their \gls{D2D} link and feed it back to the \gls{BS} that benefits from this knowledge to optimize the performance of \gls{D2D} communications. However, this feedback is imperfect since a limited  number of resources are available for the exchange of control information. Thus, in a limited feedback network, a quantized channel measurement is reported to the \gls{BS} by either quantizing the properties of the transmitted signal (e.g. modulation, beam-forming vector) or quantizing the channel (i.e. adapt the transmitted signal to the channel). \footnote{Part of this work has been presented at ACM MSWiM 2018}
\subsection{Related work} 
Resource allocation in cellular networks suffers from the imperfection in \gls{CSI} knowledge which is mainly caused by: limited resources available for feedback, channel estimation error and the channel feedback delay. The effect of feedback delay on resource allocation is studied in \cite{Ahmad2009}. Feedback allocation and resource allocation were proposed in \cite{Deghel2018} where both limited feedback resources and delayed  feedback  information  are  assumed. Based on the outdated \gls{CSI} knowledge available at the \gls{BS}, authors in \cite{Ahmad2010} proposed a centralized resource allocation and relay selection framework that reduces the power consumption of cellular networks. \cite{Hassan2011}

In addition, researchers have been interested in elaborating new resource allocation schemes that improve the performance of \gls{D2D} networks in terms of: throughput, interference management and energy saving etc. Most of the existing works assume the global \gls{CSI} knowledge at the \gls{BS} level and propose centralized \gls{D2D} resource allocation algorithms. Several tools have been used for the study of centralized resource allocation problems: stochastic geometry modeling and resource optimization (e.g. \cite{Ye2014}), centralized graph-theoretic approach (e.g. \cite{Maghsudi2016} and \cite{Zhang2013}), mixed-integer programming (e.g. \cite{Aijaz2014}, \cite{Zulhasnine2010} and \cite{Han2012}), particle swarm optimization (e.g.\cite{Su2013}), non-convex optimization problem using branch-and-bound method (e.g. \cite{Yu2014}) and coupled processors approach (e.g.\cite{Vitale2015Modeling}) etc.

Several works limit the amount of CSI overhead by considering a partial CSI knowledge of the \gls{D2D}-enabled cellular network. In \cite{Feng2016}, authors consider that the BS has the global CSI knowledge except the interference links between UEs. In \cite{Maghsudi2016}, the CSI knowledge is restricted to the cellular links. Authors of \cite{Tang2016} consider statistical (and not instantaneous) CSI and propose a power allocation scheme for D2D-underlay cellular systems based on monotonic optimization. A stochastic cutting plane algorithm was proposed in \cite{Liu2015} to achieve the cross-layer resource optimization without the knowledge of the channels' statistics.

The main challenge of centralized approaches is the need for the \gls{D2D} \gls{CSI} at the \gls{BS} level which suffers from a trade-off between the large amount of overhead (i.e. especially in scenarios where the channels vary rapidly with time) and the imperfect knowledge of the channels' states. Therefore, the full \gls{CSI} knowledge assumption is not practical and pushes for performing distributed approaches for resource allocation of \gls{D2D} communications.

Assuming the knowledge of the utility function at the \gls{D2D} users' level, game theory has been the main tool used for elaborating  distributed resource allocation: pricing (e.g. \cite{Ye2015}), auctions (e.g. \cite{Xu2013}) and coalition formation (e.g. \cite{Li2014}) etc. Game-theoretical approaches do not solve the overhead problem because users still need to share information (i.e. prices or bids etc). In addition, several works have evaluated the performance of both centralized and distributed approaches for \gls{D2D} resource allocation. Authors in \cite{Yin2015} proved that their distributed algorithm achieves interesting performance gain with significant reduction of signaling overhead. In \cite{Maghsudi2016}, both centralized and distributed resource allocation strategies were proposed for an underlay \gls{D2D} communication system.

Energy consumption consists a pertinent and important part of \gls{5G} networks at different levels: ecological side, customer satisfaction and mobile network operators' expenses. Hence, energy consumption has been one of the main performance criteria that scheduling algorithms aim to optimize. Centralized (e.g. \cite{MJung_ModeSelecion} and \cite{Xiao2011}) as well as distributed (e.g \cite{Silva2014} and \cite{Wang2015}) resource allocation algorithms were proposed for reducing the energy consumption of \gls{D2D}-enabled cellular networks. 
 
\subsection{Contribution and Organization} 
The centralized characteristic of today's scheduling in cellular networks suffers from the ignorance of the global \gls{CSI} knowledge of the network. Thus, the scheduling will always be limited by the number of resources available for \gls{CSI} reporting. This weakness will be multiplied by the use of \gls{D2D} technique where the \gls{D2D} channels are estimated at the \gls{D2D} receiver level and then reported to the \gls{BS}. From that comes the idea of having a distributed scheduling that benefits for the \gls{D2D} users' knowledge of their local \gls{CSI}. In this work, reducing the energy consumption under throughput constraint is the main goal of the proposed scheduling framework. However, we recall that the proposed algorithms are applicable for any other performance metric.

Overlay \gls{D2D}, i.e. dedicated resources for \gls{D2D} communications, is assumed in order to avoid interference. In this work, we propose both centralized and distributed scheduling algorithm that optimizes the energy consumption of overlay-\gls{D2D} networks under throughput constraints.\footnote{This is only an example and does not limit the application of our algorithm to any other \gls{D2D} performance metrics.} This optimization problem is studied based on Lyapunov technique. Lyapunov functions for general non-linear systems (i.e. especially stability analysis) is considered as a robust theoretical and practical tool. We start dealing with this Lyapunov optimization problem by proposing a centralized approach where the \gls{D2D} resource allocation is managed by the central entity, i.e. \gls{BS}. Based on channels' statistics, the \gls{BS} chooses a subset of \gls{D2D} pairs that will send their \gls{CSI} feedback to the \gls{BS}. Only the corresponding subset of \gls{CSI}s is then received at the \gls{BS} and the \gls{BS} schedules then the optimal \gls{D2D} link based on this subset of \gls{CSI}s knowledge. We show that the performance of the proposed centralized algorithm achieves that of the optimal centralized scheduling in a limited feedback network. 

In an ideal scenario, without limitation on the resources available for feedback transmission, centralized solution is the optimal one since the \gls{BS} can acquire the instantaneous \gls{CSI} knowledge of all the \gls{D2D} pairs. However, in realistic context of limited feedback scenario, the proposed centralized scheduling suffers from the limited resources available for feedback transmission where only a subset of \gls{D2D} pairs will be able to send its \gls{CSI} to the \gls{BS}. Since this subset is selected based on the statistics of  channel states, there will be no guarantee that the optimal \gls{D2D} link will be scheduled. In this work, we show that in limited feedback networks, distributed solutions may take advantage of the local \gls{CSI} knowledge of the \gls{D2D} pairs to achieve higher performance. 

Therefore, we propose a distributed algorithm that benefits from the users' knowledge of their local \gls{D2D} channel state to intelligently manage the spectrum access and optimize the energy consumption of \gls{D2D} communications under throughput constraint. Indeed, each user transmits a simple control indicator to reveal the value of its local \gls{D2D} channel state. Based on these indicators, all the users, including the optimal one, manage to report information concerning their local \gls{CSI}. However, a collision may occur while sharing these \gls{CSI} indicators. The impact of this collision is discussed and some strategies for reducing its probability are proposed. Under these conditions of collision reduction, the distributed scheduling identifies the optimal \gls{D2D} links and tends by that to achieve the performance of the \textit{ideal} scheduling (i.e. where the \gls{BS} knows the instantaneous \gls{CSI} of all the \gls{D2D} links without any cost). 

We show how both centralized and distributed algorithms can be simply implemented in real cellular network, i.e. \gls{LTE}. Numerical results reveal how the suggested algorithms reduce the energy consumption of \gls{D2D} communications under throughput constraint. In a limited feedback \gls{D2D} networks, the proposed centralized algorithm schedules \gls{D2D} communications based on the \gls{CSI} statistics whereas the proposed distributed algorithm benefits from the users' knowledge of their instantaneous local \gls{CSI}. Indeed, the distributed algorithm outperforms the centralized one due to the fact that, in a limited feedback \gls{D2D} network, \gls{D2D} users have more information about their local \gls{CSI} than the \gls{BS}. 

This paper will be organized as follows. In Section \ref{sec:RA_SystModel}, the system model and the optimization problem are described. The centralized approach is detailed in Section \ref{sec:RA_Cent}: the algorithm is exposed and its optimistic property is proved. The distributed algorithm is detailed in section \ref{sec:RA_Dist} and its performance is analyzed. Section \ref{sec:RA_ProbCollision} computes the probability of collision that may occur during the transmission of the \gls{CSI} indicators. Strategies for reducing this probability are discussed. Section \ref{sec:RA_Implementation} reveals how the proposed algorithms can be implemented in today's cellular networks. Numerical results in Section \ref{sec:RA_NumResults} show the performance of the proposed algorithms and illustrate the significant energy consumption reduction provided by the distributed algorithm compared to other scheduling policies. Section \ref{sec:RA_Conclusion} concludes the paper whereas the proofs are provided in the appendices. 

\section{System model \label{sec:RA_SystModel}}
We consider a set of $\mathrm{N}$ pairs of users that want to communicate with each other via \gls{D2D} links (see figure \ref{fig.scenario} as an example where $N=6$). Users are randomly distributed in a cell of radius $\mathrm{R_c}$ such that the distance between \gls{D2D} pair is set within the range $\left[d_{min},d_{max} \right]$. We denote by time-slot the time scale of the resource allocation decision. The channel between any two nodes in the network is modeled by a Rayleigh fading channel that remains constant during one time-slot and changes independently from one time-slot to another based on a complex Gaussian distribution with zero mean and unit variance.

A practical scenario with adaptive modulation is assumed. At each time slot, a device can support a bit-rate adapted to its channel conditions (i.e. \gls{SNR}) and selected from a set of  $\mathrm{M}$ bit rates ${\mathsf{\left[ R_1, ..., R_M\right]}}$. These bit-rates correspond to a set of $\mathrm{M}$ \gls{SNR} thresholds $\mathsf{\left[ S_1, ..., S_M\right]}$. When a \gls{D2D} link has a \gls{SNR} that lies within the interval $\left[ \mathsf{S_m},\mathsf{S_{m+1}}\right[$ then its quantized value of \gls{SNR} is equal to $\mathsf{S_m}$ and it will support a bit rate of $\mathsf{R_m}$. When the $n^{th}$ \gls{D2D} link has a \gls{SNR} of $\mathsf{S_m}$ then it can achieve a bit-rate equal to $\mathsf{R_m}$ by transmitting at a power $P_{n,m}$: 
\begin{equation}
\label{eq.Pow}
P_{n,m}:=\min \left\lbrace \frac{\mathsf{S_m}N_o}{|h_n|^2{L}_n }\,\,,\,\, P_{max} \right\rbrace 
\end{equation}
with $h_n$ the fading coefficient of the $n^{th}$ \gls{D2D} pair, $L_n$ the path-loss over the $n^{th}$ \gls{D2D} link that mainly depends on the $n^{th}$ \gls{D2D} pair distance  $d_n$, $N_o$ the noise power and $P_{max}$ the maximum user's transmission power. 
\begin{figure}[H]
\begin{centering}
\includegraphics[scale=0.25]{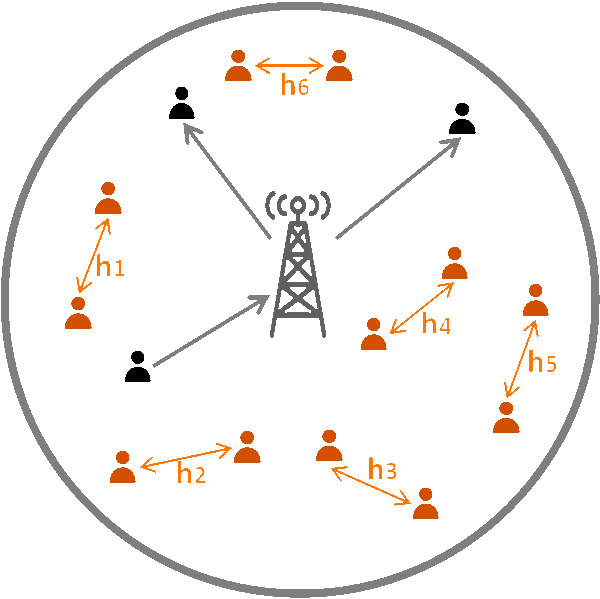}
\par\end{centering}
\caption{An example of D2D scenario where $N=6$}
\label{fig.scenario} 
\end{figure}

We consider overlay \gls{D2D} (e.g. see \cite{RitaOverlay2017}) where dedicated resources are allocated to \gls{D2D} links in order to mitigate interference between \gls{D2D} and cellular communications (i.e. no reuse of cellular resources). Scheduling scheme is the algorithm that manages the radio resources' access of \gls{D2D} communications. Assuming a user time division multiplex access scheduling, the available \gls{D2D} resources are only used by one \gls{D2D} communication at a given time-slot. We assume that, based on some pilot reference signals, each \gls{D2D} transmitter has a channel state estimation of its corresponding \gls{D2D} link. Deploying an energy aware scheduling requires the \gls{CSI} knowledge of these \gls{D2D} links. Therefore, based on existing control channels, we propose that \gls{D2D} links transmit a smart indicator of their local \gls{CSI} information in order to enable a scheduling that reduces the energy consumption of \gls{D2D} communications. $N _{RB}$ represents the number of \gls{RB} available each time-slot for \gls{D2D} \gls{CSI} reporting. We propose a new \gls{CSI} reporting technique that benefits efficiently from these limited resources in order to reduce the energy consumption of the \gls{D2D} communications while guaranteeing throughput constraints.

In the sequel, we use the following notations for a given time-slot $t$: $R_n\left( t \right)$  as the service rate of the $n^{th}$ \gls{D2D} communication, $P_n\left(t\right)$ as the transmission power of the $n^{th}$ \gls{D2D} communication and $R_{th}$ as the threshold of the time average throughput of \gls{D2D} communications with $\gamma_{th}$ the corresponding \gls{SNR}.

The aim of this work is to design a scheduling framework that reduces the energy consumption of \gls{D2D} networks under throughput constraint and in a limited feedback network where \gls{D2D} transmitters are limited by $N_{RB}$ \gls{RB}s for \gls{CSI} exchange. We define $\Gamma\left( t \right)$ as the scheduled user at time-slot $t$. Hence, the \gls{OP} below consists of finding the scheduling strategy $\Gamma\left( t \right)$ that minimizes the energy consumption of \gls{D2D} communications under throughput constraint and limited resources for \gls{CSI} feedback transmission:

\begin{equation}
\begin{aligned}
\label{OP_Gen}
& \underset{\Gamma}{\text{minimize}}
& & \lim_{T\rightarrow\infty}\sup\frac{1}{T}\sum_{t=1}^{T}\sum_{n=1}^{\mathrm{N}}\mathbb{E}\left[P_{n}\left(t\right)\right]
\end{aligned}
\end{equation}
\begin{equation*}
\begin{aligned}
& \text{s.t.}
& & \,\lim\limits_{T\rightarrow\infty}\inf\frac{1}{T}\sum\limits_{t=1}^{T}\mathbb{E}\left[ R_{n}\left(t\right)\right] \geq R_{th} \,\,\,\,\,\, \forall n, \\
& & & N_{RB} \text{: number of \gls{RB} for CSI feedback.}
\end{aligned}
\end{equation*}

We apply Lyapunov Optimization \cite{neely2010stochastic} to solve the problem above. Looking at problem's throughput constraint leads to the construction of the following virtual queues that help to meet the desired constraint:
\begin{equation}
\label{eq.Q}
Q_{n}\left(t+1\right)=\left[Q_{n}\left(t\right)-R_{n}\left(t\right)\right]^{+}+R_{th}
\end{equation}
From queuing theory \cite{neely2010stochastic}, we known that the throughput constraint of the problem \ref{OP_Gen} is equivalent to the strong stability of the virtual queues (\ref{eq.Q}). The optimization problem is transformed to a stabilization problem of the virtual queuing network while minimizing the time average of the users' transmission power. Hence, the drift-plus-penalty algorithm is used for minimizing the average power subject to network stability. When the system shifts to undesirable states, the defined Lyapunov function largely increases. Thus, scheduling actions that drift this function to the negative direction are crucial for settling the system stability. The scheduling policy $\Gamma$ aims to minimize the following expression:

\begin{equation}
\label{OP_Lyap}
\begin{aligned}
& \underset{\Gamma}{\text{minimize}}
& &\sum_{n=1}^{\mathrm{N}}V\mathbb{E}\left[P_{n}\left(t\right)\right]-Q_{n}\left(t\right)\mathbb{E}\left[R_{n}\left(t\right)\right]
\end{aligned}
\end{equation}
\begin{equation*}
\begin{aligned}
& \text{s.t.}
& & N_{RB} \text{: number of \gls{RB} for CSI feedback}
\end{aligned}
\end{equation*}

where $Q_{n}\left(t+1\right)=\left[Q_{n}\left(t\right)-R_{n}\left(t\right)\right]^{+}+R_{th}$ and $V$ is a non-negative weight that is chosen in such a way that the desired performance trade-off between the power minimization and the virtual queue size is achieved. The scheduling solution of (\ref{OP_Lyap}) aims to achieve a time average of the users' power consumption within a distance of at most $O\left(\frac{1}{V}\right)$ from the optimal value while ensuring a time average virtual queue backlog of $O\left(V\right)$.

\section{Centralized approach \label{sec:RA_Cent}}
For the centralized approach, we suppose that the user's \gls{CSI} feedback contains the following information concerning its \gls{D2D} link: (i) channel quality (i.e. transmission rate) as well as (ii) transmission power. Since we aim to minimize the \gls{D2D} users' transmission power, both channel quality and transmission power information are required. Due to the limited amount of resources available for \gls{CSI} reporting, it is not possible for all the users to transmit their \gls{CSI} feedback each time-slot neither to transmit the exact continuous values of their \gls{CSI}.  Therefore, depending on the number of resource blocks $N_{RB}$ available for \gls{CSI} reporting, a limited number of users is able to simultaneously transmit its quantized \gls{CSI} feedback to the \gls{BS}. The number of quantized \gls{CSI} feedback (e.g. 20-22 encoded bits) that can be simultaneously supported at a given time-slot $t$ is denote by $K^{\left(1\right)}\left( N_{RB}\right)$ which depends on the number $N_{RB}$ of resources available for \gls{CSI} reporting. For clarity, we omit the variable $\left( N_{RB}\right)$ when $K^{\left(1\right)}$ notation is used.

The centralized approach is based on the following three phases algorithm: \textbf{Phase 1} where the \gls{BS} chooses, based on global statistical \gls{CSI}, the subset $\Lambda^*$ of users that will transmit their \gls{CSI} feedback to the \gls{BS} (i.e. with $|\Lambda^*|\leq K^{\left(1\right)}$);  \textbf{Phase 2} where the \gls{BS} receives the \gls{CSI} feedback from the users belonging to the subset $\Lambda^*$ and \textbf{Phase 3} where the \gls{BS} schedules the user, in the subset $\Lambda^*$, that optimizes the energy consumption metric (\ref{OP_Lyap}) of the \gls{D2D} communications. We denote by $\Omega$ the set of all the possible subset of $K^{\left( 1\right)}$ different users.

\subsection{Centralized Algorithm}
Moreover, the centralized algorithm is given in Algorithm \ref{Algo_Cent} for a time-slot $t$ and will be detailed in this subsection.

\textbf{Phase 1}: The goal of this phase is to choose the subset $\Lambda^*$ of $K^{\left(1\right)}$ \gls{D2D} transmitters that will send their \gls{CSI} to the \gls{BS} at a given time-slot. Based on the global knowledge of the \gls{D2D} \gls{CSI} statistics, the \gls{BS} computes the optimal subset $\Lambda^*$ as follows:
\begin{equation}
\label{statCSI}
\resizebox{1\hsize}{!}{$
\Lambda^*:=\argmin\limits_{\Lambda \subset \varOmega}\mathbb{E}_{h}\left[\min\limits_{n \in \Lambda}\left[VP_{n}\left(t,h\right)-Q_{n}\left(t\right)R_{n}\left(t,h\right)\right]\right]$}
\end{equation}
\textbf{Phase 2}: Each transmitter $n$ of the subset $\Lambda^*$ will proceed as follows: (i) computes the index $m^* \in \lbrace 1,..,\mathrm{M} \rbrace$ that minimizes its utility function $VP_{n,m}\left(t\right)-Q_{n}\left(t\right)\mathsf{R_{m}}$, (ii) fixes respectively its bit rate and its transmission power as follows: $R_n\left(t\right)=\mathsf{R_{m^*}}$ and $P_n\left(t\right)=P_{n,m^*}\left(t\right)$ (iii) quantizes the transmission power $P_n\left(t\right)$ (i.e. denoted by $\tilde{P}_{n}\left(t\right)$) (iv) sends a \gls{CSI} feedback that contains both: the channel quality (i.e. which implies the chosen bit rate $R_n\left(t\right)$) and the quantized transmission power $\tilde{P}_{n}\left(t\right)$.

\textbf{Phase 3} Among the users in the subset $\Lambda^*$, the \gls{BS} schedules the \textit{optimal user} $n^*$ which corresponds to the user that verifies equation (\ref{eq.BUE}):

\begin{equation}
\label{eq.BUE}
n^*=\argmin_{n \in\Lambda^*}\left[V\tilde{P}_{n}\left(t\right)-Q_{n}\left(t\right)R_{n}\left(t\right)\right]
\end{equation}

\begin{algorithm}
\caption{Centralized scheduling at the level of the \gls{BS}}\label{Algo_Cent}
\begin{algorithmic}[1]
\State Finds subset $\Lambda^*$ given by (\ref{statCSI})
\State Sends \gls{CSI} reporting request to users in subset $\Lambda^*$ 
\State Receives $R_n\left(t\right)$ and $\tilde{P}_n\left(t\right)$ from all users $n \in \Lambda^*$
\State Find the optimal user $ n^*$ given by (\ref{eq.BUE})
\State Update $\bm{Q}$ based on (\ref{eq.Q})
\end{algorithmic}
\end{algorithm}
\subsection{Stability and optimistic criteria}
We prove that the proposed centralized scheduling achieves a distance of at most $O\left( \frac{1}{V} \right)$ from the optimal solution of the centralized scenario while guaranteeing the stability of the system of virtual queues.
We denote by $P_c^*$ the optimal solution of the \gls{OP} (\ref{OP_Gen}) when the best centralized scheduling is applied in a limited feedback network. In this case, the \gls{BS} knows the global statistical \gls{CSI} and the instantaneous \gls{CSI} of only a subset of $K^{\left(1\right)}$ users.
\begin{proposition}
\label{prop.StabCentQeuues}
When the proposed centralized algorithm is applied, the total average backlogs of the queues is upper bounded by a finite value $\frac{C+B}{\epsilon}$:
\begin{equation}
\limsup\limits_ {T \rightarrow \infty} \frac{1}{T}\sum\limits_{t=0}^{T-1} \sum\limits_{i=1}^\mathrm{N}\mathbb{E}\left[Q_i\left(t\right) \right] \leq \frac{C+B}{\epsilon}
\end{equation}
 \end{proposition}
The proposed centralized scheduling policy ensures the strong stability of the virtual queuing network with an average queue backlog of $O\left(V\right)$. Hence, the throughput constraint of the \gls{OP} (\ref{OP_Gen}) is satisfied.

\begin{proposition}
\label{prop.Cent_bounds}
For the centralized approach, the time average of power consumption verifies the following:
\begin{equation}
\resizebox{1\hsize}{!}{$
{P}_c^{*} \leq \lim_{T\rightarrow\infty}\sup\frac{1}{T}\sum\limits_{t=1}^{T}\sum\limits_{n=1}^{\mathrm{N}}\mathbb{E}\left[P_{n}\left(\Gamma^{cent}\left(t\right)\right)\right] \leq {P}_c^{*}+\frac{C}{V}$}
\label{eq.boundPowCent}
\end{equation}
Where $C$ and $V$ are finite and the value of $V$ is tuned in such a way that the time average power is as close as possible to the solution of the optimal centralized limited-feedback scenario $P_c^*$ with a corresponding queue size trade-off.
\end{proposition}

\begin{proof}
Proofs of propositions \ref{prop.StabCentQeuues} and \ref{prop.Cent_bounds} are based on Lyapunov technique and detailed in Appendix-\ref{proof.Dist_bounds}.
\end{proof}

We deduce that for a large finite value of $V$, the proposed centralized algorithm achieves the optimal solution of the centralized scenario which has $P_c^*$ as the optimal time average power.

\section{Distributed approach \label{sec:RA_Dist}}
Finding the optimal solution of problem (\ref{OP_Lyap}) requires the global \gls{CSI} knowledge of \gls{D2D} links. This knowledge is limited by the restricted number $N_{RB}$ of resources available for \gls{CSI} reporting. However, one can profit from the local \gls{CSI} knowledge at the \gls{D2D} users' level in order to propose a new way for handling the resources available for \gls{CSI} feedback. This enables all the users to feedback some indicators concerning their \gls{CSI} at each time-slot.  After sharing \gls{CSI} indicators, the user that optimizes (\ref{OP_Lyap}) is identified. This study shows that the proposed distributed approach largely reduces the energy consumption of the \gls{D2D} network. The proposed distributed scheduling (given in Algorithm \ref{Algo_Dist}) can be summarized as follows:\\
\textbf{Phase  1}: each \gls{D2D} pair estimates its channel state in order to compute its energy consumption metric. Then, each \gls{D2D} user shares a simple \gls{CSI} indicator (e.g. 1 or 2 encoded bits per time-slot) in such a way that the index of the \gls{RE} on which this \gls{CSI} indicator is transmitted point out the value of the energy consumption metric.\\
\textbf{Phase 2}: The \gls{D2D} user that has transmitted its \gls{CSI} indicator at the \gls{RE} of the lowest index and which corresponds to the user that minimizes the energy consumption of \gls{D2D} communications is scheduled for data transmission. \\
\textbf{Phase 3}: This phase aims to reduce the collision that may occur during the transmission of \gls{CSI} indicators. 

\subsection{Distributed algorithm}
The different steps of the distributed algorithm (given in algorithm \ref{Algo_Dist}) are detailed in the sequel. The number of \gls{CSI} indicators (e.g. 1 or 2 encoded bits) that can be simultaneously supported at a given time-slot $t$ is denote by $K^{\left(2\right)}\left( N_{RB}\right)$ which depends on the number $N_{RB}$ of resources available for \gls{CSI} reporting. For clarity, we omit the variable $\left( N_{RB}\right)$ when $K^{\left(2\right)}$ notation is used.

\textbf{Phase 1}: Based on pilot reference signals, at each time-slot $t$, the $n^{th}$ \gls{D2D} pair estimates its \gls{D2D} channel state $h_n \left(t \right)$ and deduces its energy consumption metric given by equation (\ref{eq.vn}).
\begin{equation} 
\label{eq.vn} 
v_n\left(t\right)=\min_{m \in \lbrace 1,...,\mathrm{M} \rbrace}{VP_{n,m}\left(t\right)-Q_n\left(t\right)R_m\left(t\right)}
\end{equation}
 
We can verify that the values of the utility function $v_n\left(t\right)$ fit within the range $\left[v_{min}\left(t\right),v_{max}\left(t\right) \right]$ given by equations (\ref{eq.vmin}) and  (\ref{eq.vmax}) where $r=1$ and $f=0$. 

\begin{equation} 
\label{eq.vmin} 
v_{min}\left(t\right)=-tR_{th}\mathsf{R_M}
\end{equation}
\begin{equation} 
\label{eq.vmax} 
v_{max}\left(t\right)=v_{min}\left(t\right)+r\frac{VP_{max}-R_{th}\mathsf{R_1}-v_{min}\left(t\right)}{\left( K^{\left(2\right)}\right)^f}
\end{equation}
These border variables are identically computed by each device in such a way that all the devices will have the same values of $v_{min}\left(t\right)$ and $v_{max}\left(t\right)$. The interval  $\left[v_{min}\left(t\right),v_{max}\left(t\right) \right]$ serves for the discretization of the utility function $v_n\left(t\right)$ into $K^{\left(2\right)}$ equal intervals. The values of $v_n\left( t\right)$ from the continuous set $\left[v_{min}\left(t\right),v_{max}\left(t\right) \right]$ are mapped to a finite set $\mathbb{S}_v$ of $K^{\left(2\right)}$ elements (see equation (\ref{eq.set_v})). The simple way to quantize the utility function $v_{n}$ is to choose the closest element to $v_n$ within $\mathbb{S}_v$. 
\[
\mathbb{S}_v=\bigcup\limits_{ j=1,...,K^{\left(2\right)}} a_j
\]
\begin{equation}
\label{eq.set_v} =\bigcup\limits_{ j=1,...,K^{\left(2\right)}} \left\{ v_{min}+ \left( j-1 \right)\frac{v_{max}-v_{min}}{K^{\left(2\right)}-1}\right\}
\end{equation}

\begin{algorithm}[H]
\caption{Distributed scheduling at the level of each user $n$}\label{Algo_Dist}
\begin{algorithmic}[1]
\State Receives, from \gls{BS}, constants $R_{th}$, $T_p$ and $V$ (given by \ref{eq.V_eps}) 
\For{$1 \leq t \leq T_p$}
\State Estimates channel state $h_n\left(t\right)$
\State Computes performance metric $v_n\left(t\right)$ from (\ref{eq.vn})
\State Computes $v_{min}$ and $v_{max}$ from (\ref{eq.vmin}) and (\ref{eq.vmax})
\State Finds $\lbrace \tilde{v}_{n},\tilde{k}_n \rbrace$ given by (\ref{eq.matching_n}) and (\ref{eq.matching_k})
\State Shares CSI indicator at the $\tilde{k}_n^{th}$ RE 
\If {Collision}
\State Detects collision index $c$ 
\State Updates the values of $r$ and $f$ from (\ref{eq.Update_r_f})
\EndIf
\If {$n==n^*$ given by (\ref{eq.BUE_dist})}
\State Transmits data to its \gls{D2D} pair 
\EndIf
\State Update $Q_n\left(t\right)$ from (\ref{eq.Q})
\EndFor
\end{algorithmic}
\end{algorithm}
In practice, cellular networks contain limited resources for \gls{CSI} reporting (i.e. here denoted by $N_{RB}$ resource blocks). Thus, instead of transmitting the quantized \gls{CSI} value, each \gls{D2D} pair limits its feedback to a simple \gls{CSI} indicator (e.g. 1 or 2 encoded bit) that is sufficient for describing its utility function. This mechanism, i.e. called \textbf{\textit{Channel Indexing Feedback}}, consists of introducing a mapping between the quantized value of the utility function (i.e. within the $\mathbb{S}_v$ set) and the $K^{\left(2\right)}$ \gls{RE}s available for the transmission of the \gls{CSI} indicators. In other terms, the index of the \gls{RE}s used for the \gls{CSI} indicator's transmission is sufficient for identifying the corresponding quantized value of the energy consumption metric. Figure \ref{fig.Matching} illustrates an example where the proposed {Channel Indexing Feedback} technique is applied on 4 different users with $n=1,2,3$ and $4$. For the $n^{th}$ \gls{D2D} pair, the expressions of the quantized utility function (i.e. denoted by $\tilde{v}_n$) and its corresponding mapped \gls{RE} used for the transmission of the \gls{CSI} indicator (i.e. denoted by $\tilde{k}_n$) are formally given by: 

\begin{equation}
\label{eq.matching_n}
\tilde{v}_n\left(t\right)=\argmin \limits_{x_i \in \mathbb{S}_{v}}  \left( v_n\left(t\right)-x_i \right)  \mathbbm{1}_{\lbrace v_n\left(t\right)>x_i \rbrace}
\end{equation}
\begin{equation}
\label{eq.matching_k}
\tilde{k}_n\left(t\right)=\argmin \limits_{ i\in \left\lbrace  1,..., K^{\left(2\right)} \right\rbrace  }  \left( v_n\left(t\right)-x_i \right) \mathbbm{1}_{\lbrace v_n\left(t\right)>x_i \rbrace}
\end{equation}
where $\mathbbm{1}$ represents the indicator function.
\begin{figure}
\begin{centering}
\includegraphics[scale=0.6]{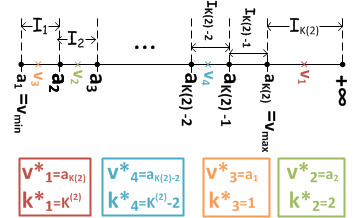}
\par\end{centering}
\caption{Example of channel indexing feedback for 4 users}
\label{fig.Matching} 
\end{figure}

In phase 1, at each time-slot $t$, each transmitter $n$ computes its performance metric $v_n\left(t\right)$ given by (\ref{eq.vn}) and then applies {\textit{Channel Indexing Feedback}} technique to compute the couple $\left \lbrace \tilde{v}_n\left(t\right), \tilde{k}_n\left(t\right) \right \rbrace$ given respectively by (\ref{eq.matching_n}) and (\ref{eq.matching_k}) and which correspond to the quantized value of $v_n\left(t\right)$ and its corresponding index within the set $\mathbb{S}_v$. The transmitter $n$ sends its \gls{CSI} indicator at the $\tilde{k}_n^{th}$  \gls{RE} among the $K^{\left(2\right)}$ available \gls{RE}s. 

In phase 2, depending on the shared \gls{CSI} indicators by all the \gls{D2D} pairs,  the scheduled user is the one that exclusively transmits its \gls{CSI} indicator on the \gls{RE} with the lowest index $\tilde{k}_n\left(t\right)$. The chosen user corresponds to the \gls{D2D} pair that minimizes the utility function $v_n\left(t\right)$ at time-slot $t$ and by that reduces the energy consumption of \gls{D2D} network. The index $n^*$ of the scheduled \gls{D2D} pair is given by:
\begin{equation}
\label{eq.BUE_dist}
\resizebox{1\hsize}{!}{$n^*=\argmin\limits_{n \in \lbrace 1,...,N \rbrace }\lbrace \tilde{v}_n \,\,\,\, | \,\, \tilde{v}_m \neq \tilde{v}_n  \,\,\,\,\forall \,\, m \in \lbrace 1,...,N \rbrace / n \rbrace$}
\end{equation}

The proposed algorithm is distributed in the sense that \gls{D2D} users are responsible of computing their utility function as well as transmitting the corresponding \gls{CSI} indicators based on their estimated \gls{CSI}. The computation efforts are highly reduced at the \gls{BS} level. However, recognizing the optimal \gls{D2D} pair remains an open question since it can be done in a fully distributed manner or in a \gls{BS}-assisted way. In the sequel, we discuss these two approaches:
\begin{itemize}
\item \textbf{ \gls{BS}-assistance:} \gls{BS} is responsible of listening to the transmitted \gls{CSI} indicators by all the \gls{D2D} pairs in order to detect the optimal user whose \gls{CSI} indicator has been transmitted at the \gls{RE} of the lowest index. Then, the \gls{BS} schedules the identified optimal user. Such network assistance encounters the business challenge of distributed algorithm for mobile network operators that generally prefer to support centralized administrated solutions for controlling the network and guaranteeing the performance of cellular communications. 
\item \textbf{ Autonomous:} Supposing that devices have full duplex capacities, they can simultaneously transmit and receive the \gls{CSI} indicators. Hence, each user will be able to autonomously recognize whether it corresponds to the optimal user that should be scheduled or not. If a \gls{D2D} pair send its \gls{CSI} indicator on \gls{RE} of index $c$ and does not receive any \gls{CSI} indicators on the \gls{RE}s of index $<c$ then this \gls{D2D} pair will recognize that it will be scheduled. This algorithm has important benefit for autonomous networks (i.e. without any centralized entity) but faces important security issues that need to be solved.
\end{itemize}
In phase 3, during the transmission of \gls{CSI} indicators in phase 1, a collision may occur when at least two users transmit their \gls{CSI} indicators at the same \gls{RE} (i.e. at least two users have the same quantized performance metric $\tilde{v}_n$). Phase 3 consists of applying some strategies that reduce the occurrence probability of such collisions. These procedure are detailed in section \ref{sec:RA_ProbCollision}.

\subsection{Performance analysis}
We denote by ${P}_{id}^{*}$ the time average of the power when an \textit{ideal} scheduling is considered in the sense that the network has the global instantaneous \gls{CSI} knowledge of the \gls{D2D} links. In this case, the \textit{ideal} scheduling achieves the optimum of problem (\ref{OP_Gen}) without any constraints on the number of resources available for \gls{CSI} feedback. When a collision free scenario is considered, we prove that the proposed distributed scheduling achieves a distance of at most $O\left( \frac{1}{V} \right)$ from the \textit{ideal} solution while ensuring the strong stability of the virtual queuing network with an average queue backlog of $O\left(V\right)$.

\begin{proposition}
\label{prop.Dist_bounds}
Assuming that no collision occurs; the distributed scheduling $\Gamma^{dist}\left( t\right)$ guarantees a time average power consumption that verifies the following:
\begin{equation}
\resizebox{1\hsize}{!}{${P}_{id}^{*} \leq \lim\limits_{T\rightarrow\infty}\sup\frac{1}{T}\sum\limits_{t=1}^{T}\sum\limits_{n=1}^{\mathrm{N}} \mathbb{E}\left[ P_{n}\left(\Gamma^{dist}\left(t\right)\right) \right] \leq {P}_{id}^{*}+\frac{C}{V}$}
\label{eq.boundPowDist}
\end{equation}
Where $C$ and $V$ are finite and the value of $V$ is tuned to make the time average power as close as desired to the ideal solution with a corresponding virtual queue size trade-off.
\end{proposition}

\begin{proof}
See Appendix-\ref{proof.Dist_bounds} based on Lyapunov technique.
\end{proof}

We deduce that for a tuned finite value of $V$, the proposed distributed algorithm achieves the performance of the \textit{ideal} solution $P_{id}^*$ of the \gls{OP} (\ref{OP_Gen}). with an error of $O\left( \frac{1}{V} \right)$. For high finite values of $V$, the distributed scheduling reaches the \textit{ideal} performance when no collision occurs at the transmission level of the \gls{CSI} indicators. The process of reducing this collision are discussed in the coming section. 
\begin{remark}
In the centralized approach, the feedback sent by \gls{D2D} user $n$  to the \gls{BS} contains both its transmission power $\tilde{P}_n\left(t\right)$ and rate $R_n\left(t\right)$ at a given time-slot $t$. For the centralized algorithm, the \gls{BS} needs to acquire the values of these two information (rate and power) in order to update the value of the virtual queues $\bm{Q}\left(t\right)$ and deduce by that the subset of \gls{D2D} users $\Lambda^*$ that will send its feedback to the \gls{BS} at the next time-slot. However, in the distributed approach, only the value of the utility metric $v_n\left(t\right)$ of each \gls{D2D} user $n$ ($v_n\left(t\right)=VP_n\left(t\right)-Q_n\left(t\right)R_n\left(t\right)$) is required at a given time-slot $t$. The actual value of $P_n\left(t\right)$ and $R_n\left(t\right)$ are not necessary since each \gls{D2D} user can locally update the value of its virtual queue.
\end{remark}
\section{Probability of collision \label{sec:RA_ProbCollision}}
During the phase 2 of the distributed algorithm, the transmission of the \gls{CSI} indicators may suffer from a collision. It is crucial to note that this collision occurs at the level of the \gls{CSI} indicators' transmission and not at the level of the data transmission. A collision takes place when at least two users have the same quantized utility function $\tilde{v}_n$ and thus transmit their \gls{CSI} indicators at the same \gls{RE}. In particular, when a collision occurs at the level of the \gls{RE} of index $k^*$, identifying the optimal user to schedule is not possible anymore. We define the {overall collision} as the scenario where each user collides at least with another one; hence none of the users is scheduled at this time-slot. In order to avoid such scenario, two precautions detailed in the sequel, were adopted:
\begin{itemize}
\item The Lyapunov constant $V$ is chosen based on equation (\ref{eq.V_eps}) in order to minimize the collision probability.
\item The mapping, that matches the discrete values of the energy consumption metric of \gls{D2D} users with the $K^{\left(2\right)}$ available \gls{RE}, is updated as in equation (\ref{eq.Update_r_f}) in the aim of avoiding future collisions.
\end{itemize}
\subsection{Value of Lyapunov constant}
We describe how to limit the probability of feedback collision by choosing the appropriate value of the Lyapunov constant $V$. We consider that the value of the $V$ is updated within a periodicity of $T_p$ time-slots. We start by finding the analytic expression of the collision probability. A collision occurs at a given \gls{RE} when at least two users transmit their \gls{CSI} indicators at the same \gls{RE}. We call \textbf{probability of collision} $P_{c}$ as the probability of occurring an overall collision event where each \gls{CSI} indicator transmission collides with at least another one in such a way that none of the users is scheduled (i.e. none of the users has exclusively transmitted its \gls{CSI} indicator on one of the available \gls{RE}s).  In this part, we limit the analytic result to the single bit-rate case $\mathrm{M}=1$ (i.e. corresponding to bit-rate $\mathsf{R}$ and \gls{SNR} $\mathsf{S}$).
\begin{proposition}
\label{prop.Prob_Colli}
The probability of collision $P_c$ is given by:
\begin{equation}
\label{eq.Prob_Collision}
P_{c}=1-\sum_{i=1}^{\mathrm{N}}\sum_{j=1}^{K^{\left(2\right)}}\bar{p}_c \lbrace i,j \rbrace \prod_{k =1}^{j-1}\left(1-\sum_{l=1 \neq i}^{\mathrm{N}}\bar{p}_c \lbrace l,k \rbrace \right)
\end{equation}
\[\text{where }\bar{p}_c \lbrace i,j \rbrace=2\left[\exp\left(c_{i,j-1} \right)-\exp\left( c_{i,j}\right)\right] 
\]
\[
\times\prod_{k=1\neq i}^{\mathrm{N}}\left[1-2\exp\left( c_{k,j-1}\right)+2\exp\left( c_{k,j}\right)\right]
\]
\[
\text{and } c_{i,j}=-\frac{V\mathsf{S}N_{o}}{\left(a_{j}+Q_{i}\mathsf{R}\right)L_i}
\]
\end{proposition}
\begin{proof}
See Appendix-\ref{proof.Prob_Colli}.
\end{proof}
Based on the expression (\ref{eq.Prob_Collision}), we can tune the value of the Lyapunov constant $V$ in order to limit the overall collision probability to a small $\epsilon$.
  \begin{theorem}
 \label{eq.Th_PC2_prime}
The probability $P_c$ is bounded by a given $\epsilon$ (with $0 \leq \epsilon \leq 1$) when the value of the Lyapunov constant $V$ is given by:\\
\begin{equation}
\label{eq.V_eps}
 V\left( \epsilon \right)=-\frac{R_{th}\mathsf{R}\ln\left(\epsilon'\right)T_{p}}{P_{max}\ln\left(\epsilon'\right)+\mathsf{S} N_{0}L_{min}^{-1}}
\end{equation}
\[\text{where } \epsilon':=\frac{1}{2\mathrm{N}}\left[1-\left(\frac{1-\epsilon}{\mathrm{N}K^{\left(2\right)}}\right)^{\frac{1}{\mathrm{N}+K^{\left(2\right)}}}\right]
\]
and $L_{min}$ is the path-loss over a \gls{D2D} link with $d_{min}$ as the smallest distance allowed between a \gls{D2D} pair.
 \end{theorem}
 
\begin{proof}
See Appendix-\ref{proof.Th_PC2_prime}.
\end{proof}
\subsection{Updating the mapping}
If a collision occurs then the granularity of the mapping set $\mathbb{S}_v$ is not sufficient for proposing different \gls{CSI} indicators to describe the energy consumption metric of the different \gls{D2D} pairs. Hence, an accuracy improvement of the mapping set $\mathbb{S}_v$ is done in order to avoid future collisions. We denote by $c$ the smallest \gls{RE}'s index where a collision has occurred (with $1\leq c \leq K^{\left(2\right)}$ ). The mapping $\mathbb{S}_v$ between the quantized value of the energy consumption metric and the $K^{\left(2\right)}$ available \gls{RE}s is updated by modifying the parameters $r$ and $f$ of the $v_{max}$ formula (\ref{eq.vmax}) as follows: 
\begin{equation}
\label{eq.Update_r_f}
\begin{aligned}
& r=c {\text{ ; if  }\left( r<K^{\left(2\right)}\right)} &  f=f+1\\
& \text{else  } \left( r==K^{\left(2\right)}\right) & f=0
\end{aligned}
\end{equation}
When collision occurs at the \gls{RE} of index $c$ with $c < K^{\left(2\right)}$ then the mapping update aims to reduce the probability of collision by taking the following actions: (i) $v_{max}=a_c$ to reduce the interval $\left[v_{min},v_{max}\right]$ and (ii) $f=f+1$ to increase the granularity of the intervals within the new subset $\mathbb{S}_v$. However, when collision occurs at the \gls{RE} with the highest index ($c=K^{\left(2\right)}$), this means that all the \gls{D2D} users have a utility function higher than the current $v_{max}$. Thus, the mapping update aims to reduce the probability of collision by enlarging the interval $\left[v_{min},v_{max}\right]$ (with $r=K^{\left(2\right)}$ and $f=0$).

\section{Implementation \label{sec:RA_Implementation}}
We show how the proposed centralized and distributed approaches can be implemented in real networks. To do so, we consider the example of \gls{LTE} network. We focus on how the existing \gls{PUCCH} formats can be modified in order to support the proposed algorithms. 

\subsection{Existing feedback Standardization} 
In this work, we benefit from the existing \gls{DL} feedback standards in \gls{LTE} specifications developed by \gls{3GPP} (i.e. one can refer to \cite{TS36213} for more details). The \gls{PUCCH} appears mainly in two formats depending on the type of the handled information: (i) the formats 1,1a,1b of 1-2 encoded bits and which are dedicated for the ACK/NACK feedback and (ii) the formats 2,2a,2b of 20-22 encoded bits and which are mainly used for \gls{CSI} feedback. The \gls{CSI} feedback consists of three component: (i) the \gls{RI}, the \gls{PMI} and the \gls{CQI}. The preferred triplet \gls{CSI} (\gls{RI}/\gls{PMI}/\gls{CQI}) is computed by each user based on its instantaneous channel estimations obtained from \gls{DL} pilot. In the aim of reducing the complexity of such computation, several algorithms have been proposed in the literature (e.g. \cite{Singh2013}).

As mentioned before, the \gls{BS} may improve the performance of \gls{D2D} communications by acquiring the \gls{CSI} of \gls{D2D} links. The reporting of these \gls{CSI} can be performed in two different ways: (i) periodic \gls{CSI} report (summarized version and economical in terms of radio resources) and (ii) aperiodic \gls{CSI} report (detailed version and costly in terms of radio resources). By default, the periodic \gls{CSI} reporting is done on the \gls{PUCCH}. However, when the user's data are planned to be transmitted on the \gls{PUSCH} then the \gls{CSI} reporting is multiplexed with the data and sent on the \gls{PUSCH}. Meanwhile, the aperiodic \gls{CSI} reporting is exclusively transmitted on the \gls{PUSCH} after the reception of a \gls{BS} request via a \gls{PDCCH} that carries a \gls{DL} \gls{DCI} of format 0. 

The periodic \gls{CSI} feedback, sent on \gls{PUCCH} resources, are configured based on a semi-static scheduling. This configuration is specifically assigned to each \gls{UE} via radio resource control \gls{RRC} in order to avoid the need of increasing the size of \gls{PDCCH}. When semi-static scheduling is deployed, the \gls{BS} pre-configures each user with a given resource allocation identifier and periodicity. The limited amount of \gls{PUCCH} resources prohibits the transmission of all the users' \gls{PUCCH} at each time-slot (i.e. refers to \gls{TTI} in \gls{LTE}) and obliges the possible \gls{CSI} reporting to include only the necessary information and not the detailed one. This limitation motivates us to propose a new way to manage these critical \gls{PUCCH} resources and enable by that an energy aware scheduling. 

The performance of \gls{DL} communications depends on the precision of the \gls{CSI} feedback which is function of the frequency granularity. However, the more accurate the \gls{CSI} feedback is the more the \gls{UL} feedback load is important. Different \gls{CSI} reporting mode on \gls{PUCCH} resources are defined as function of the trade-off existing between the \gls{DL} performance and the \gls{UL} load: (i) mode 1: wideband \gls{CQI} report where a single \gls{CQI} value corresponds to the entire system bandwidth and (ii) mode 2: subband \gls{CQI} report where the system bandwidth is divided into multiple subbands with different \gls{CQI} value for each subband. In this work, we limit the implementation section to the case of wideband \gls{CQI} report. Please note that considering the case of subband \gls{CSI} reporting is a straightforward process (i.e. in this case, channel indexing feedback technique is applied respectively for each subband \gls{CSI}).

Cellular networks are designed based on a centralized approach where the \gls{BS} presents the entity that controls the operations and guarantees the quality of cellular communications. The \gls{BS} will start by configuring the sub-frames corresponding to the \gls{CSI} feedback via the identification of the bandwidth of the \gls{PUCCH} region, the period of feedback and the cyclic shifting (i.e. permitting the time multiplexing between the \gls{CSI} reports of different \gls{UE}s and/or between the \gls{CSI} reports of the same \gls{UE}). Depending on these parameters, each user transmits periodically its \gls{CSI} feedback. However, due to the restricted number of \gls{PUCCH} resources, only a subset of users, i.e. not the totality of the users, will transmit their \gls{CSI} feedback at a given \gls{TTI}. Depending on the received \gls{CSI} feedback, the \gls{BS} runs its scheduling algorithm in order to allocate \gls{D2D} resources. This is the baseline protocol to which we will compare the proposed scheduling algorithms.

\subsection{Centralized algorithm implementation}
Recall that we consider user scheduling in such a way that only the optimal user is scheduled at a given \gls{TTI} to transmit its data overall the available \gls{D2D} resources. For the centralized approach, we suppose that the users send their \gls{CSI} by the use of \gls{PUCCH} format 2b. The control information contained in this \gls{PUCCH} format (\gls{CQI} and 2-bits for ACK or NACK) will be modified as follows: \gls{CQI} will remain intact however the 2-bits of ACK-NACK will be used for indicating the \gls{D2D} users' transmission. Thus, users' transmission power will be quantized by mapping their continuous values to a countable small set of 4 elements: $\left\lbrace \tilde{P}_1, \tilde{P}_2, \tilde{P}_3, \tilde{P}_4\right\rbrace$. Based on these modifications, the three phases centralized algorithm (detailed in algorithm \ref{Algo_Cent}) can be implemented.

When the \gls{PUCCH} format 2b is adopted for \gls{CSI} feedback, then based on \cite{bouguen2012lte} we can deduce the number $K^{\left(1\right)}$ of  \gls{CSI} feedback that can be simultaneously supported. As shown in figure \ref{fig.PUCCH_RB},           $K^{\left(1\right)}$ is equal to the product of the two following identifiers (i.e. communicated by the \gls{BS} via \gls{RRC}):
\begin{itemize}
\item $N_{RB}$ that indicates the number of \gls{RB}s that are available per \gls{TTI} for \gls{CSI} feedback of a \gls{PUCCH} 2/2a/2b format. It can be a configurable parameter such that the \gls{BS} can control the \gls{UL} bandwidth and can eventually dimension the size of these resources depending on the need.
\item $C^{\left(1\right)}_{MUX}$ that indicates the multiplexing capacity per \gls{RB}. This corresponds to the number of users that can send their \gls{CSI} feedback on the same \gls{RB}. Giving that this parameter depends only on the cyclic shifting of the base sequences (which is a fix value in this case) then $C^{\left(1\right)}_{MUX}=12$.
\end{itemize}
We deduce that $K^{\left(1\right)}=N_{RB} \times C^{\left(1\right)}_{MUX} = 12N_{RB}$.
\begin{figure}[ptb]
\begin{centering}
\includegraphics[scale=0.4]{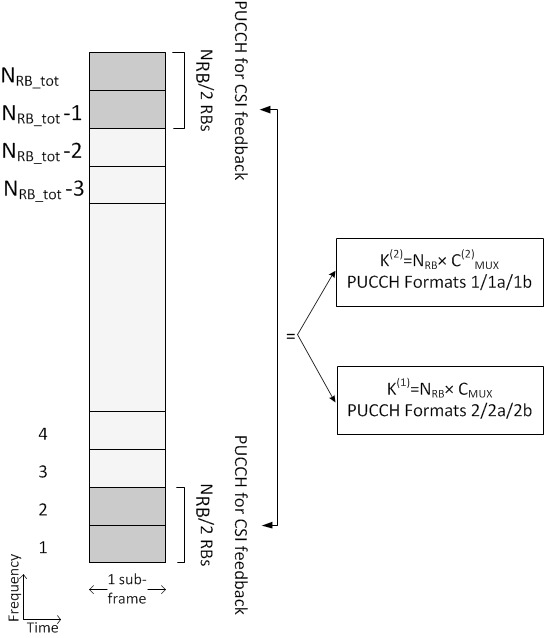}
\par\end{centering}
\caption{The resource blocks allocated for users' feedback (for both \gls{PUCCH} Format 1 and 2)}
\label{fig.PUCCH_RB} 
\end{figure}

\subsection{Distributed algorithm implementation}
In the distributed approach, the resources available for \gls{CSI} feedback are handled in a new distinct way in order to guarantee that the optimal user, at each \gls{TTI}, sends its \gls{CSI} indicator. Based on existing standards, we consider that the \gls{CSI} indicators consists of \gls{PUCCH} format 1/1a/1b which are commonly used for ACK/NACK messages or scheduling request. These control information correspond to 1 or 2 encoded bits per \gls{TTI}. Indeed, users transmit their \gls{CSI} indicators as \gls{PUCCH} format 1/1a/1b on the \gls{RE} that corresponds to their quantized energy consumption metric.

Based on \cite{bouguen2012lte}, we can deduce the number $K^{\left(2\right)}$ of \gls{PUCCH} format 1/1a/1b that can be simultaneously supported at a given time-slot and thus can be used for the {\textit{Channel Indexing Feedback}} technique.  As shown in figure \ref{fig.PUCCH_RB}, $K^{\left(2\right)}$ is equal to the product of the two following identifiers communicated by the \gls{BS} via \gls{RRC} (i.e. $K^{\left(2\right)}=N_{RB}\times C_{MUX}^{\left(2\right)}$):
 
\begin{itemize}
\item $N_{RB}$: indicates the number of \gls{RB}s available per \gls{TTI} for the transmission of \gls{PUCCH} format 1/1a/1b. It can be a configurable parameter such that the \gls{BS} can control the \gls{UL} bandwidth and eventually dimension the size of these resources depending on the need.
\item $C_{MUX}^{\left(2\right)}$ as the multiplexing capacity per \gls{RB}, which means the number of \gls{PUCCH} format 1/1a/1b that can be transmitted on the same \gls{RB}. This parameter depends on both parameters: (i) the number of possible orthogonal codes $N_{OC}$ (e.g. three for normal cyclic prefix and two for extended cyclic prefix) and (ii) the difference $\Delta_{shift}^{PUCCH}$ between two consecutive cyclic shifting for resources using the same orthogonal code. Then, $C_{MUX}^{\left(2\right)}=12N_{OC} / \Delta_{shift}^{PUCCH}$.
\end{itemize}

Using \gls{PUCCH} format 1/1a/1b, the procedure of the distributed scheduling given in Algorithm \ref{Algo_Dist} is described as follows. The \gls{BS} initiates the scheduling by announcing the constants used by the users to compute their energy consumption metric (e.g. the Lyapunov constant $V$, update duration $T_{p}$, throughput threshold $R_{th}$). At each time-slot $t$, each user $n$ computes the couple $P_n\left(t\right)$ and $R_n\left(t\right)$ that minimizes its utility function $v_n\left(t\right)$ based on its instantaneous estimated \gls{D2D} channel $h_n\left(t\right)$. Applying equations (\ref{eq.vmin}) and (\ref{eq.vmax}), the users can locally compute the values of $v_{min}\left(t\right)$ and $v_{max}\left(t\right)$ in order to deduce the mapping set $\mathbb{S}_v$. Each \gls{D2D} user deduces the quantized value of its utility function $\tilde{v}_{n}$ from equation (\ref{eq.matching_n}) as well as its corresponding \gls{RE} index $\tilde{k}_n$ from equation (\ref{eq.matching_k}). Each \gls{D2D} user $n$ sends a \gls{CSI} indicator, as a \gls{PUCCH} format 1b, on the $\tilde{k}_n^{th}$ \gls{RE}. Considering \gls{BS}-assistance approach of the distributed algorithm, \gls{BS} decides to schedule the user that exclusively transmits its \gls{CSI} indicator at the \gls{RE} of the lowest index. This is equivalent to scheduling the user that minimizes the utility function $v_n \left(t \right)$ at time-slot $t$ and minimizes by that the energy consumption of \gls{D2D} network. If a feedback collision occurs at some level of the \gls{RE}s, then the \gls{BS} transmits the smallest index $c$ of \gls{RE}s where collision has occurred. The broadcasting of these scheduling information can be done via \gls{PDCCH} that carries a \gls{DCI} of formats 1A or 1C. Depending on the collision index received from the \gls{BS}, users update their parameters $r$ and $f$ as in equation (\ref{eq.Update_r_f}) in order to increase the granularity of the intervals within the set $\mathbb{S}_v$ and reduce by that the probability of future collision. 

\section{Numerical Results \label{sec:RA_NumResults}}

We summarize the numerical settings of this section in the table \ref{Table.Settings}. The performance are evaluated by averaging over $100$ different \gls{UE}s realizations. We estimate the value of the Lyapunov constant $V$ based on equation (\ref{eq.V_eps}) and we find $V=10^{15}$ for $\epsilon=0.1$, $\mathsf{S}=80$ dB,  $\mathsf{R}=700$ kbps/RB and $R_{th}=500$ kbps/RB for $\gamma_{th}=14$dB. In \gls{LTE}, there are 15 different values for \gls{CQI} (i.e. mapping between \gls{CQI} and modulation etc...). Hence, we suppose the existence of $15$ different bit-rates that could be applied for link adaptation model (i.e. $\mathrm{M}=15$). From an internal link-level simulator we deduce a throughput-\gls{SNR} mapping for a $10\,$MHz \gls{E-UTRA} \gls{TDD} network. This mapping gives practical values of the bit-rates $\lbrace \mathsf{R}_1, ...,  \mathsf{R}_{\mathrm{15}}\rbrace$ as well as their corresponding \gls{SNR} values $\lbrace \mathsf{S}_1, ...,  \mathsf{S}_{\mathrm{15}}\rbrace=\lbrace 0, 1, ...,  14\rbrace $ dB. 
\begin{figure}[H]
\begin{centering}
\includegraphics[scale=0.8]{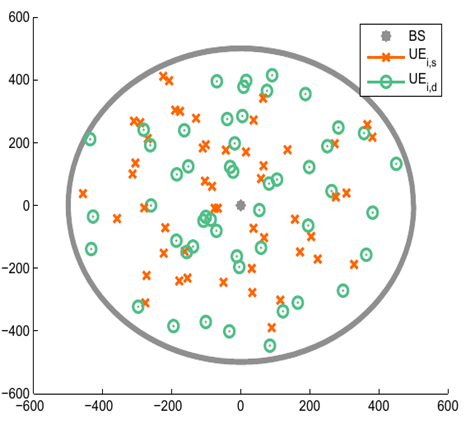}
\par\end{centering}
\caption{Uniform Random Localization of $N=50$ pairs of \gls{D2D}}
\label{fig.NR_Scenario} 
\end{figure}

\begin{table}[H]
\begin{center}
\begin{tabular}{|c|c|}
\hline
\textbf{Parameter} & \textbf{Value} \\
\hline
{Cell Radius $R_c$} & {500m} \\
\hline
{Bandwidth} & {10MHz} equivalent to $50$ \gls{RB} \\
\hline
{\gls{UE} drop} & $N=50$ \gls{UE}s\\
& Random Uniform drop with \\
 & $d_{min}=3$ m , d$_{max}=350$ m \\
&  from \cite{3GPP_ProSe} and \cite{Qualcomm} \\
&  e.g. Fig. \ref{fig.NR_Scenario}\\
\hline
Feedback  &  $N_{RB}=2$, $\Delta_{shift}^{PUCCH}=1$, $N_{OC}=3$  \\
Parameters  &  $\Rightarrow K^{\left(1\right)}=24$ and $K^{\left(2\right)}=72$  \\
\hline
 $P_{max}$ & $250$ mW  \\
\hline
Quantized  &  $\tilde{P}_1=50$, $\tilde{P}_2=100$ mW \\
Powers  &  $\tilde{P}_3=150$, $\tilde{P}_4=200$ mW \\
\hline
Path-loss & outdoor-to-outdoor path-loss in\\
& Channel Models section of \cite{3GPP_ProSe}\\
\hline
Simu. Settings  &  $100$ realizations of $T_p=10^6$ ms each\\
\hline
Noise density & $-174$ dBm/Hz\\
\hline
\gls{D2D} Noise Figure & $9$ dB\\
\hline
\end{tabular}
\caption{Numerical Settings}
\label{Table.Settings}
\end{center}
\end{table}

$\mathrm{N}$ \gls{D2D} pairs are uniformly distributed in a cell of radius $R_c$. The scheduling scheme determines how these \gls{D2D} communications access the \gls{D2D} resources in the network. The performance of the proposed algorithms is evaluated by comparing the time average of the users' energy consumption and energy efficiency between the following different algorithms: 
\begin{itemize}
\item \textbf{Centralized-limited feedback} scheduling: proposed in section \ref{sec:RA_Cent}.
\item \textbf{Distributed} scheduling: proposed in section \ref{sec:RA_Dist}.
\item \textbf{Ideal} scheduling: \gls{BS} has the global knowledge of the instantaneous channel states of all the \gls{D2D} links.
\item \textbf{Round-Robbin} scheduling: each subset $\Lambda$ of users is chosen in equal portions of time and in a circular order for the transmission of their \gls{CSI} feedback using \gls{PUCCH} format 2/2a/2b. The number of users that can send simultaneously their \gls{CSI} feedback $K^{\left(1\right)}$ depends on the number of resource blocks available for feedback transmission. 
\end{itemize}
Fig. \ref{fig.logEC} shows how the proposed distributed and centralized approaches reduce the time average of the users' transmitted power compared to the Round-Robbin approach for different \gls{SNR} thresholds $0 \leq \gamma_{th} \leq 14$ dB. The reduction is up to $70\%$ for the centralized approach and up to $98\%$ for the distributed approach. The proposed distributed algorithm outperforms the centralized one that suffers from a limited number of users that simultaneously transmit their \gls{CSI} feedback to the \gls{BS}.The distributed algorithm proposes a new way to manage the limited resources available for feedback transmission. Thus, all the \gls{D2D} users benefit from their local \gls{CSI} knowledge and limit their feedback transmission to a small \gls{CSI} indicator. This new channel indexing feedback technique guarantees the scheduling of the optimal user. Nevertheless, the distributed scheme does not achieve the ideal one as collision may occur. Even though a collision probability of $0.1$ occurs, the distributed algorithm highly reduces the  users' transmission power.

In the aim of verifying that the throughput constraint is ensured while minimizing the users' transmission power, we study the \gls{EE} of the proposed algorithms. The \gls{EE} metric is defined as the ratio of the total throughput to the total transmitted power over all the simulation duration (see \cite{Chen2010}). Figure \ref{fig.EE} represents the evolution of the \gls{EE} of the proposed algorithms as function of the \gls{SNR} threshold $\gamma_{th}$. These results show an important enhancement of the network \gls{EE} and underline the performance of the distributed algorithm compared to the other non \textit{ideal} scheduling.

\begin{figure}
\begin{centering}
\includegraphics[scale=0.4]{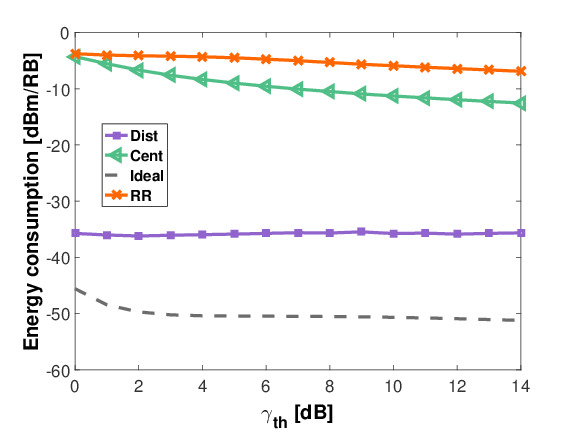}
\par\end{centering}
\caption{EC as function of the \gls{SNR} threshold $\gamma_{th}$}
\label{fig.logEC} 
\end{figure}

\begin{figure}
\begin{centering}
\includegraphics[scale=0.4]{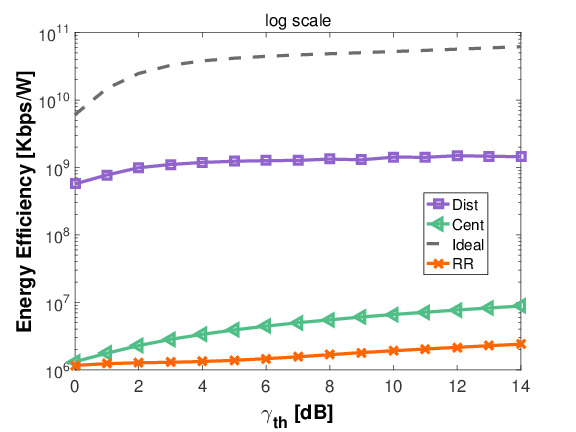}
\par\end{centering}
\caption{EE as function of the \gls{SNR} threshold $\gamma_{th}$}
\label{fig.EE} 
\end{figure}

\section{Conclusion\label{sec:RA_Conclusion}}
In paper, we consider \gls{D2D}-enabled cellular network with limited feedback and we propose a channel adaptive resource allocation algorithms that minimize the \gls{D2D} transmission power under throughput constraints. Based on Lyapunov optimization, we start by considering a centralized channel adaptive scheduling where \gls{D2D} users send their \acrfull{CSI} and transmission power information to the base station. Since limited resources are available for \gls{CSI} transmission, only a subset of \gls{D2D} users are able to report their \gls{CSI} at a given time-slot and the performance of centralized approaches is limited by the resources available for feedback exchange. Therefore, we develop a distributed approach where each \gls{D2D} user benefits from its local channel state knowledge to share a short \gls{CSI} indicator that is sufficient for deducing its corresponding energy consumption metric. Some mechanisms are deployed to avoid the collision that may occur during the \gls{CSI} indicators transmission. When feedback collision is avoided, the distributed approach achieves energy consumption performance comparable to the ideal solution where the base station has a global \gls{CSI} knowledge of all the \gls{D2D} users. Furthermore, we have discussed how both distributed and centralized algorithms can be simply implemented in existing \acrfull{LTE} standards. Numerical results show that the distributed algorithm outperforms the centralized one. Compared to a round-robin scheduling, the proposed algorithms reduce the energy consumption  between $70\%$ to $98\%$ and enhance the energy efficiency of \gls{D2D} networks.

\appendix
\subsection*{A. Proof of proposition \ref{prop.StabCentQeuues} and \ref{prop.Cent_bounds} \label{proof.Cent_bounds}}
We denote by $\Gamma^{cent}$ the proposed centralized scheduling and by $\Gamma_{c}^*$ the \textit{best} centralized scheduling that optimizes OP (\ref{OP_Gen}) in a limited-feedback \gls{D2D} network. In this centralized scenario, we assume that the \gls{BS}  has the global statistical knowledge of the \gls{D2D} \gls{CSI} as well as the instantaneous \gls{CSI} of only a subset of \gls{D2D} users. The number of users in this subset $K^{\left(1\right)}$ depends on the number of resources  available for \gls{CSI} reporting.

Lyapunov technique (see \cite{GergiadisNeely2006}) is applied for proving proposition \ref{prop.Cent_bounds}. This statement evaluates the performance of the centralized algorithm $\Gamma^{cent}$ compared to that of the \textit{best} centralized limited-feedback scheduling $\Gamma_c^{*}$  in terms of time average of the users' power consumption. In order to solve the optimization problem (\ref{OP_Gen}), the virtual queues  (\ref{eq.Q}) were introduced:
\[
Q_{n}\left(t+1\right)=\left(Q_{n}\left(t\right)-R_{n}\left(t\right)\right)^{+}+R_{th} 
\]

From queuing theory \cite{neely2010stochastic}, it is known that the strong stability of the queues means that the average arrival rate is less than the average departure rate. Thus, we can establish the equivalence between the throughput constraint of the optimization problem and the strong stability of the system of virtual queues. Therefore, our goal is to minimize the energy consumption of \gls{D2D} communications while guaranteeing the stability of the virtual queues. We base our proof on Lyapunov optimization theory and define the Lyapunov function $L_{\bm{Q}}\left( t \right)$ as follows:
\[
L_{\bm{Q}}\left( t \right) :=\frac{1}{2} \sum\limits_{i=1}^\mathrm{N}Q_i^2\left(t\right)
\]
We denote by $\bm{Q}\left( t \right)=\left( Q_1\left( t\right), ..., Q_N\left( t\right) \right)$  the vector of the current virtual queue backlogs and by $\bm{H}\left(t\right)=\left[h_1, h_2, ..., h_N \right]$ the vector of channel states. Applying the methodology in \cite{GergiadisNeely2006} gives that the drift-plus-penalty expression is upper bounded by the following:
\[
\Delta \left( \bm{Q}\left( t \right)\right) :=\Delta \left( \bm{L}\left( t \right)\right) + V\sum\limits_{i=1}^\mathrm{N}  \mathbb{E}\left[ P_i\left( t\right)| {\bm{Q}}\left( t \right)\right] 
\]
\[ \resizebox{1\hsize}{!}{$
\leq C+\sum\limits_{i=1}^{\mathrm{N}}Q_{i}\left(t\right)R_{th} -\mathbb{E}\left[\sum\limits_{i=1}^{\mathrm{N}}Q_{i}\left(t\right)R_{i}\left(\Gamma^{cent}\left(t\right),\bm{H}\left(t\right)\right)|{\bm{Q}}\left( t \right)\right]$}
 \]
\begin{equation}
\label{eq.LyapDrift_cent}
+V\mathbb{E}\left[\sum_{i=1}^{\mathrm{N}}P_{i}\left(\Gamma^{cent}\left(t\right),\bm{H}\left(t\right)\right)|{\bm{Q}}\left( t \right)\right]
\end{equation}
 where $\Gamma^{cent}\left(t \right)$ is the proposed centralized scheduling policy and C is a finite constant such that 
 \[
 \mathbb{E}\left[\sum \limits_{i=1}^{\mathrm{N}}\frac{R_{th}^{2}+R_{i}^{2}\left(t\right)}{2}|{\bm{Q}}\left( t \right)\right] \leq C\]

At each time-slot $t$, the proposed algorithm uses the global statistical \gls{CSI} to compute the subset $\Lambda^*$  of users that will simultaneously transmit their \gls{CSI} feedback to the \gls{BS}. The scheduled user is the one that has the lowest value of the metric $v_i\left(t \right)$. Based on equation (\ref{statCSI}) of the chosen subset $\Lambda^*$, the proposed centralized algorithm $\Gamma^{cent}\left( t\right)$ verifies the inequality (\ref{eq.AvgComp_cent}) for any scheduling policy $\Gamma\left( t\right)$ including the centralized optimal policy $\Gamma_c^*\left( t\right)$ of \gls{OP} (\ref{OP_Gen}).
\[
\resizebox{1\hsize}{!}{$\sum\limits_{i=1}^{\mathrm{N}}\left[VP_{i}\left(\Gamma^{cent}\left(t\right),\bm{H}\left(t\right)\right)-Q_{i}\left(t\right)R_{i}\left(\Gamma^{cent}\left(t\right),\bm{H}\left(t\right)\right)\right]|\bm{Q}\left( t \right) $}
\]
\begin{equation}
\label{eq.AvgComp_cent}
\resizebox{1\hsize}{!}{$\leq \sum\limits_{i=1}^{\mathrm{N}}\left[VP_{i}\left(\Gamma\left(t\right),\bm{H}\left(t\right)\right)-Q_{i}\left(t\right)R_{i}\left(\Gamma\left(t\right),\bm{H}\left(t\right)\right)\right]|\bm{Q}\left( t \right)$}
\end{equation}
Since the optimal policy $\Gamma_{c}^*\left(t\right)$ verifies the OP (\ref{OP_Gen}), this policy satisfies the corresponding throughput constraint and guarantees by that the stability of the virtual queues. Supposing that  the arrival rate vector of these virtual queues is interior to the stability region of the system of queue, then for an $\epsilon >0$ we have:
\[
\mathbb{E}\left[R_i\left( \Gamma_{c}^*\left(t\right)\right)| {\bm{Q}}\left( t \right) \right]=\mathbb{E}\left[R_i\left( \Gamma_{c}^*\left(t\right)\right) \right] \geq R_{th} + \epsilon
\] 
 Same methodology as in \cite{GergiadisNeely2006} gives the following upper bound of the total average backlogs of the virtual queues:
\begin{equation}
\label{eq;StabilityQueuesDist}
\limsup\limits_ {T \rightarrow \infty} \frac{1}{T}\sum\limits_{t=0}^{T-1} \sum\limits_{i=1}^\mathrm{N}\mathbb{E}\left[Q_i\left(t\right) \right] \leq \frac{C+B}{\epsilon}
\end{equation}
 with $B$ a finite constant such that:
 \[
  \limsup\limits_ {T \rightarrow \infty} \frac{1}{T}\sum\limits_{t=0}^{T-1}\sum\limits_{i=1}^{\mathrm{N}} \mathbb{E}\left[VP_i\left( \Gamma_{c}^*\left(t\right)\right) \right]\leq B
  \] 
Thus, proposition \ref{prop.StabCentQeuues} is verified and all the virtual queues in the system are strongly stable. Hence, when the arrival rate at the virtual queues is less than its average departure rate then the proposed centralized scheduling satisfies the throughput constraint in (\ref{OP_Gen}). Pursuing with the Lyapunov optimization of queuing networks leads to the proposition \ref{prop.Cent_bounds} as follows:
 \[
\frac{1}{T}\sum\limits_{t=0}^{T-1} \sum\limits_{i=1}^\mathrm{N}\mathbb{E}\left[P_i\left(\Gamma^{cent}\left(t\right)\right) \right]   \leq \frac{C}{V}+\frac{\mathbb{E}\left[L_{\bm{Q}}\left(0\right) \right]}{VT} 
\]
\[+\frac{1}{T}\sum\limits_{t=0}^{T-1} \sum\limits_{i=1}^\mathrm{N}\mathbb{E}\left[P_i\left(\Gamma_c^{*}\left(t\right)\right) \right]
\]

Therefore, the proposed centralized algorithm achieves the performance of the \textit{best} solution of the limited-feedback problem (\ref{OP_Gen}) with a distance of $O\left(\frac{1}{V}\right)$ and a time average queue backlog of $O\left({V}\right)$. 

\subsection*{B. Proof of proposition \ref {prop.Dist_bounds} \label{proof.Dist_bounds}}
We denote by $\Gamma^{dist}$ the proposed distributed scheduling and by $\Gamma_{id}^*$ the \textit{ideal} scheduling policy that optimizes problem (\ref{OP_Gen}) while assuming the global knowledge of the instantaneous \gls{CSI} of \gls{D2D} links. Lyapunov technique (see \cite{GergiadisNeely2006}) is applied for verifying proposition \ref{prop.Dist_bounds} that evaluates the performance of the distributed algorithm $\Gamma^{dist}$ compared to that of the \textit{ideal} scheduling in terms of time average of the users' power consumption.
Applying the methodology in \cite{GergiadisNeely2006} gives that the drift-plus-penalty expression is upper bounded by the following:
\[ \Delta \left( \bm{Q}\left( t \right)\right)\leq C+\sum\limits_{i=1}^{\mathrm{N}}Q_{i}\left(t\right)R_{th}
\] 
\[+V\mathbb{E}\left[\sum_{i=1}^{\mathrm{N}}P_{i}\left(\Gamma^{dist}\left(t\right),\bm{H}\left(t\right)\right)|{\bm{Q}}\left( t \right)\right]
 \]
\begin{equation}
\label{eq.LyapDrift}-\mathbb{E}\left[\sum\limits_{i=1}^{\mathrm{N}}Q_{i}\left(t\right)R_{i}\left(\Gamma^{dist}\left(t\right),\bm{H}\left(t\right)\right)|{\bm{Q}}\left( t \right)\right]
\end{equation}
 where $\Gamma^{dist}\left(t \right)$ is the proposed distributed scheduling policy and C is a finite constant such that 
 \[
 \mathbb{E}\left[\sum \limits_{i=1}^{\mathrm{N}}\frac{R_{th}^{2}+R_{i}^{2}\left(t\right)}{2}|{\bm{Q}}\left( t \right)\right] \leq C\].

  In order to compare between the proposed distributed scheduling and the \textit{ideal} one, we recall the respective procedure of these policies. For the \textit{ideal} scheduling $\Gamma_{id}^*$, we suppose the global \gls{CSI} knowledge of the \gls{D2D} communications. Based on these \gls{D2D} channel states, both the transmission power and the throughput of all the \gls{D2D} links are recognized. Thus, the scheduled user corresponds to the one that optimizes problem (\ref{OP_Gen}) (i.e. the one that minimizes the consumption power under the throughput constraint).

For the proposed distributed scheduling, at each time-slot $t$, each \gls{D2D} user $i$ estimates its channel state and deduces its energy metric $v_i\left(t\right)$. Due to an existing  mapping between the discrete values of the energy consumption metric and the $K^{\left(2\right)}$ available resource elements for \gls{CSI} feedback; each user $i$ sends a simple \gls{CSI} indicator on the \gls{RE} that maps with its discrete value $\tilde{v}_i\left(t\right)$. Supposing a non collision scenario, the scheduled user corresponds to the one that transmits its \gls{CSI} indicator on the \gls{RE} with the lowest index (i.e. the user that has the lowest value of the energy metric $v_i\left(t\right)$).  Thus, compared to any other scheduling policy $\Gamma\left(t\right)$, the proposed distributed scheduling $\Gamma^{dist}\left(t\right)$ verifies the following at each time slot $t$:
\[\resizebox{1\hsize}{!}{$
\sum\limits_{i=1}^{\mathrm{N}}\left[VP_{i}\left(\Gamma^{dist}\left(t\right),\bm{H}\left(t\right)\right)-Q_{i}\left(t\right)R_{i}\left(\Gamma^{dist}\left(t\right),\bm{H}\left(t\right)\right)\right]|\bm{Q}\left( t \right) $}
\]
\[ \resizebox{1\hsize}{!}{$\leq \sum\limits_{i=1}^{\mathrm{N}}\left[VP_{i}\left(\Gamma\left(t\right),\bm{H}\left(t\right)\right)-Q_{i}\left(t\right)R_{i}\left(\Gamma\left(t\right),\bm{H}\left(t\right)\right)\right]|\bm{Q}\left( t \right)$}
\]

The equation above is verified for any scheduling policy $\Gamma\left(t\right)$ including the \textit{ideal} policy $\Gamma_{id}^*\left(t\right)$ where the network has the global \gls{CSI} knowledge of \gls{D2D} communications. Thus:
\[\resizebox{1\hsize}{!}{$
\mathbb{E}\left [\sum\limits_{i=1}^{\mathrm{N}}\left[VP_{i}\left(\Gamma^{dist}\left(t\right),\bm{H}\left(t\right)\right)-Q_{i}\left(t\right)R_{i}\left(\Gamma^{dist}\left(t\right),\bm{H}\left(t\right)\right)\right]|\bm{Q}\left( t \right)\right] $}
\]
\begin{equation}
\label{eq.AvgComp}
\resizebox{1\hsize}{!}{$
 \leq \mathbb{E}\left [\sum\limits_{i=1}^{\mathrm{N}}\left[VP_{i}\left(\Gamma_{id}^*\left(t\right),\bm{H}\left(t\right)\right)-Q_{i}\left(t\right)R_{i}\left(\Gamma_{id}^*\left(t\right),\bm{H}\left(t\right)\right)\right]|\bm{Q}\left( t \right)\right] $}
\end{equation}

Since the ideal policy $\Gamma_{id}^*\left(t\right)$ verifies the OP (\ref{OP_Gen}), this policy verifies the corresponding throughput constraint and guarantees by that the stability of the virtual queues. Supposing that the arrival rate vector of these virtual queues is interior to the stability region of the system of queue, then for an $\epsilon >0$ we have:
\[
\mathbb{E}\left[R_i\left( \Gamma_{id}^*\left(t\right)\right)| {\bm{Q}}\left( t \right) \right]=\mathbb{E}\left[R_i\left( \Gamma_{id}^*\left(t\right)\right) \right] \geq R_{th} + \epsilon
\] 
 Same methodology as in \cite{GergiadisNeely2006} gives the following upper bound of the total average backlogs of the virtual queues:
\begin{equation}
\label{eq;StabilityQueuesDist}
\limsup\limits_ {T \rightarrow \infty} \frac{1}{T}\sum\limits_{t=0}^{T-1} \sum\limits_{i=1}^\mathrm{N}\mathbb{E}\left[Q_i\left(t\right) \right] \leq \frac{C+B}{\epsilon}
\end{equation}
 with $B$ a finite constant such that $ \limsup\limits_ {T \rightarrow \infty} \frac{1}{T}\sum\limits_{t=0}^{T-1}\sum\limits_{i=1}^{\mathrm{N}} \mathbb{E}\left[VP_i\left( \Gamma_{id}^*\left(t\right)\right) \right]$ $\leq B$.
 
Thus, all the virtual queues in the system are strongly stable. Hence, when the arrival rate at the virtual queues is less than its average departure rate then the distributed scheduling policy satisfies the throughput constraint in (\ref{OP_Gen}). Pursuing with the Lyapunov optimization of queuing networks leads to the proposition \ref{prop.Dist_bounds} as follows:
 \[
\frac{1}{T}\sum\limits_{t=0}^{T-1} \sum\limits_{i=1}^\mathrm{N}\mathbb{E}\left[P_i\left(\Gamma^{dist}\left(t\right)\right) \right]   \leq \frac{C}{V}+\frac{\mathbb{E}\left[L_{\bm{Q}}\left(0\right) \right]}{VT} 
\]
\[+\frac{1}{T}\sum\limits_{t=0}^{T-1} \sum\limits_{i=1}^\mathrm{N}\mathbb{E}\left[P_i\left(\Gamma_{id}^{*}\left(t\right)\right) \right]
\]

Therefore, we propose an algorithm based on a distributed approach that benefits from the \gls{D2D} local knowledge of their \gls{CSI} values. When collision is bypassed, the proposed distributed algorithm achieves the performance of the \textit{ideal} solution of (\ref{OP_Gen}) with a distance of $O\left(\frac{1}{V}\right)$ and a time average queue backlog of $O\left({V}\right)$. 

\subsection*{C. Proof of proposition \ref{prop.Prob_Colli} \label{proof.Prob_Colli}}
Recall that $K^{\left( 2\right)}$ \gls{RE}s are available for the transmission of the \gls{CSI} indicators.
A collision occurs at the level of the $k^{th}$ \gls{RE} when at least two users have the same quantized energy consumption metric (e.g. $\tilde{v}_i=\tilde{v}_j=a_k$ with $i \neq j$ hence a collision occurs at the $k^{th}$ element). Therefore, the probability of overall collision $P_c$ corresponds to the probability that none of the \gls{D2D} users in the network is scheduled (i.e. the \gls{CSI} indicator's transmission of each \gls{D2D} user collides with at least another one). We compute this probability of collision $P_c $ assuming a single bit rate model ($\mathrm{M}=1$) with $\mathsf{R}$ and  $\mathsf{S}$ respectively corresponding to the bit-rate and \gls{SNR} threshold.

For each element $a_j \in \mathbb{S}_v$ (with $1\leq j\leq K^{\left( 2\right)}$), we define the following two events whose probabilities of occurrence  are computed in the sequel:
\begin{itemize}
\item $A_{i,j}$: $i^{th}$ \gls{D2D} link has  $\tilde{v}_i = a_j$ 
\item $B_{i,j}$: $i^{th}$ \gls{D2D} link has $\tilde{v}_i  \geq a_{\min \left \lbrace j +1,K^{\left( 2\right)} \right \rbrace}$
\end{itemize} 

A Rayleigh fading channel $h_i$ with zero mean and unit variance is considered, hence the squared magnitude $|h_i|^2$ has an exponential distribution of parameter one. Therefore, we can deduce the probability of the two above events. For the simplification of coming expressions, we use the following notation: 
\[
c_{i,j}=-\frac{V\mathsf{S}N_{o}}{\left(a_{j}+Q_{i}\mathsf{R}\right)L_{i}}
\]
\begin{enumerate}
\item {\textbf{The probability of $A_{i,j}$ (for $1\leq j\leq K^{\left( 2\right)}$)}}
\[
\mathbb{P}\left(A_{i,j}\right)=\mathbb{P}\left(\tilde{v}_{i}=a_{j}\right)=\mathbb{P}\left(v_{i}\epsilon\left] a_{j},a_{j+1}\right]\right)
\]
\[
=\mathbb{P}\left(a_{j}< \frac{V\mathsf{S} N_{o}}{|h_{i}|^{2}L_i}-Q_{i}\mathsf{R} \leq a_{j+1}\right)
\]
\[
\resizebox{0.9\hsize}{!}{$=2\left[\exp\left(-\frac{V\mathsf{S}N_{o}}{\left(a_{j}+Q_{i}\mathsf{R}\right)L_i}\right)-\exp\left(-\frac{V\mathsf{S}N_{o}}{\left(a_{j+1}+Q_{i} \mathsf{R}\right)L_i}\right)\right]$}
\]
\[
=2\left[\exp\left( c_{i,j}\right)-\exp\left( c_{i,j+1}\right)\right]
\]
where $a_{K^{\left( 2\right)}+1}$ is equal to $+\infty$.\\
\item \textbf{The probability of $B_{i,j}$ (for $1\leq j\leq K^{\left( 2\right)}-1$)}
\[
\mathbb{P}\left(B_{i,j}\right)=\mathbb{P}\left( \tilde{v}_i  \geq a_{j+1} \right)
\]
\[
=1-2\exp\left(-\frac{V\mathsf{S}N_{o}}{\left(a_{j +1}+Q_{i}\mathsf{R}\right)L_i}\right)=1-2\exp\left(c_{i,j+1}\right)
\]
\end{enumerate}
$\bar{p}_c \lbrace i,j \rbrace$ denotes the probability that only the user $i$ has its quantized value $\tilde{v}_i$ equals to $a_j \in \mathbb{S}_v$ for $1 \leq j \leq K^{\left(2\right)}$ and is given by:
\[\bar{p}_c \lbrace i,j \rbrace=\mathbb{P}\left(A_{i,j}\right)\prod_{k=1\neq i}^{\mathrm{N}}\mathbb{P}\left(\bar{A}_{k,j}\right)\]
\[
=2\left[\exp\left( c_{i,j-1}\right)-\exp\left( c_{i,j}\right)\right]
\]
\[
\times \prod_{k=1\neq i}^{N}\left[1-2\exp\left( c_{k,j-1}\right)+2\exp\left( c_{k,j}\right)\right]
\]
The probability of collision $P_c$ is deduced as follows:
\[
P_{c}=\mathbb{P}\lbrace \text{no scheduled D2D user}\rbrace
\]
\[
 =1-\sum_{i=1}^{\mathrm{N}}\mathbb{P}\lbrace \text{D2D user $i$ scheduled}\rbrace
\]\[
=1-\sum_{i=1}^{\mathrm{N}}\sum_{j=1}^{K^{\left(2\right)}}\mathbb{P}\lbrace \text{only the $i^{th}$ D2D user has } \tilde{v}_i=a_j\rbrace
\]
\[
 \times \mathbb{P}\lbrace \text{$j=$ lowest RE index without collision}\rbrace  \]
\[
P_{c}=1-\sum_{i=1}^{\mathrm{N}}\sum_{j=1}^{K^{\left(2\right)}}\bar{p}_c \lbrace i,j \rbrace \prod_{k=1}^{j-1}\left(1-\sum_{l=1 \neq i}^{\mathrm{N}}\bar{p}_c \lbrace l,k \rbrace \right)
\]

\subsection*{D. Proof of theorem \ref{eq.Th_PC2_prime} \label{proof.Th_PC2_prime}}
The proof is split into 4 steps. In \textbf{step 1}, we start by expressing the lower bounds of $\mathbb{P}\left(A_{i,j}\right)$ and $\mathbb{P}\left(\bar{A}_{i,j}\right)$. These bounds are used, in \textbf{step 2}, for computing the upper and the lower bounds of the probability $\bar{p}_c \lbrace i,j \rbrace$ (i.e. probability that only the \gls{D2D} user $i$ transmits its \gls{CSI} indicator on \gls{RE} $j$). Therefore, in \textbf{step 3}, we find the upper bound of the collision probability $P_c$. In \textbf{step 4}, we conclude the value of the Lyapunov constant $V$ that reduces the collision probability $P_c$ to $\epsilon$.

\subsubsection*{\textbf{Step 1}} \textbf{Verify that for all $\left \lbrace i, j\right \rbrace$: }
\[
\mathbb{P}\left(A{}_{i,j}\right) \geq \mathbb{P}\left(B_{i,K^{\left(2\right)}-1}\right) 
\]
and 
\[
\mathbb{P}\left(\bar{A}_{i,j}\right)\geq  \mathbb{P}\left(B_{i,K^{\left(2\right)}-1}\right)
\]
For $\,\, 1\leq j \leq K^{\left(2\right)}-1$, we have the two following expressions:
\[\resizebox{1\hsize}{!}{$
\mathbb{P}\left(A{}_{i,j}\right)=2\left[\exp\left(\frac{-VSN_{0}L_{i}^{-1}}{a_{j}+Q_{i}R}\right)-\exp\left(\frac{-VSN_{0}L_{i}^{-1}}{a_{j+1}+Q_{i}R}\right)\right]$}
\] 

\[
\resizebox{1\hsize}{!}{$\mathbb{P}\left(A{}_{i,K^{\left(2\right)}}\right)=\mathbb{P}\left(B_{i,K^{\left(2\right)}-1}\right)=\left[1-2\exp\left(\frac{-VSN_{0}L_{i}^{-1}}{a_{k^{\left(2\right)}}+Q_{i}R}\right)\right]$}
\]
Then, $\mathbb{P}\left(A{}_{i,j}\right) \geq \mathbb{P}\left(B{}_{i,K^{\left(2\right)}-1}\right)$ and $\mathbb{P}\left(B{}_{i,j}\right) \geq \mathbb{P}\left(B{}_{i,K^{\left(2\right)}-1}\right) \, \forall 0 \leq i \leq N \, , \, \forall 1 \leq j \leq K^{\left(2\right)}$.

Based on the definition of $\mathbb{P}\left(A{}_{i,j}\right)$ and $\mathbb{P}\left(B{}_{i,j}\right)$ one can see that:
\[
\mathbb{P}\left(A{}_{i,j}\right)=\mathbb{P}\left(B{}_{i,j-1}\right)-\mathbb{P}\left(B{}_{i,j}\right) 
\]
\[
\Rightarrow \mathbb{P}\left(\bar{A}{}_{i,j}\right)=1-\mathbb{P}\left(B{}_{i,j-1}\right)+\mathbb{P}\left(B{}_{i,j}\right)
\]
\[
\Rightarrow \mathbb{P}\left(\bar{A}{}_{i,j}\right) \geq \mathbb{P}\left(B{}_{i,j}\right) \geq \mathbb{P}\left(B{}_{i,K^{\left(2\right)}-1}\right)
\]

\subsubsection*{\textbf{Step 2}} \textbf{Upper and lower bounds of  $\bar{p}_c \lbrace i,j \rbrace$:} Based on the lower bound of $\mathbb{P}\left(A{}_{i,j}\right)$ and $\mathbb{P}\left(\bar{A}_{i,j}\right)$ we deduce that:
\begin{itemize}
\item Lower bound: $
\bar{p}_c \lbrace i,j \rbrace \geq \prod\limits_{i=i}^\mathrm{N} \mathbb{P}\left(B{}_{i,K^{\left(2\right)}}\right)
$
\item Upper bound: $\bar{p}_c \lbrace i,j \rbrace  \leq \mathbb{P}\left(\bar{B}_{i+1\bmod\mathrm{N} ,K^{\left(2 \right)}}\right)$
\end{itemize}

\subsubsection*{\textbf{Step 3}} \textbf{Upper bound of $P_c$:} Based on the upper and lower bound of $\bar{p}_c \lbrace i,j \rbrace$, we deduce that:
\[
P_{c}=1-\sum_{i=1}^{\mathrm{N}}\sum_{j=1}^{K^{\left(2\right)}}\bar{p}_{c}\left\{ i,j\right\} \prod_{1}^{j-1}\left(1-\sum_{l=1\neq i}^{\mathrm{N}}\bar{p}_{c}\left\{ l,k\right\} \right)
\]
\[
P_{c}\leq1-\left[\prod\limits_{m=1}^{\mathrm{N}}\mathbb{P}\left(B_{m,K^{\left(2\right)}-1}\right)\right] 
\]
\[
\times \sum\limits_{i=1}^{\mathrm{N}}\sum\limits_{j=1}^{K^{\left(2\right)}}\left[1-\sum\limits_{l=1\neq i}^{\mathrm{N}}\mathbb{P}\left(\bar{B}_{l+1 \bmod \mathrm{N},K^{\left(2\right)}-1}\right)\right]^{j}
\]
\[
 P_{c}\leq1-\mathrm{N}K^{\left(2\right)}\prod \limits_{m=1}^{\mathrm{N}}\left[1-2\exp\left(\frac{-V\mathsf{S} N_{0}L_m^{-1}}{a_{K^{\left(2\right)}}+Q_{m}\mathsf{R}}\right)\right] 
 \]
 \[
 \times\left[1-2\sum\limits_{l=1}^{\mathrm{N}}\exp\left(\frac{-V\mathsf{S}N_{0}L_l^{-1}}{a_{K^{\left(2\right)}}+Q_{l}\mathsf{R}}\right)\right]^{K^{\left(2\right)}}
\]
For all $1 \leq i\leq N,$ we can verify that:
\[
\exp\left(\frac{-V\mathsf{S}N_{0}L_{l}^{-1}}{a_{K^{\left(2\right)}}+Q_{l}\mathsf{R}}\right)\leq\exp\left(\frac{-V\mathsf{S}N_{0}L_{min}^{-1}}{VP_{max}+R_{th}\mathsf{R}t_{simu}}\right)
\]
where $L_{min}$ is the path-loss over a \gls{D2D} link with a peer distance equals to the minimum \gls{D2D} peer distance (i.e. $d_{min}$). Hence, 
\[\resizebox{1\hsize}{!}{$
P_{c}\leq1-\mathrm{N} K^{\left(2\right)}\left[1-2\mathrm{N} \exp\left(\frac{-V\mathsf{S}N_{0}L_{min}^{-1}}{VP_{max}+R_{th}\mathsf{R}t_{simu}}\right)\right]^{\mathrm{N}+K^{\left(2\right)}}$}
\]
\subsubsection*{\textbf{Step 4}} \textbf{Deduce $V\left(\epsilon\right)$ :} In order to bound the collision probability $P_c$ by $\epsilon$, the Lyapunov constant $V\left( \epsilon \right)$  should verify theorem \ref{eq.Th_PC2_prime}.
\vspace{-10pt}

\ifCLASSOPTIONcaptionsoff
  \newpage
\fi

\bibliographystyle{IEEEtran}
\bibliography{thesis_bibliography}

\ifCLASSOPTIONcaptionsoff
  \newpage
\fi



%
%

\begin{IEEEbiographynophoto}{Rita Ibrahim}
received the B.E. degree from Telecom-Paristech, France, and Lebanese University, Lebanon, in 2015, and the M.Sc. degree in advanced communication networks jointly from Ecole Polytechnique and Telecom-Paristech, France, in 2015. She obtained a PhD degree from CentraleSupélec, France, in 2019. She is currently a research engineer in Orange Labs, France. Her current research interests include analytic modeling and performance evaluation of cellular networks, stochastic network optimization, and radio access technologies for 5G networks.
\end{IEEEbiographynophoto}
\begin{IEEEbiographynophoto}{Mohamad Assaad}
received the MSc and PhD degrees (with high honors), both in telecommunications, from Telecom ParisTech, Paris, France, in 2002 and 2006, respectively. Since 2006, he has been with the Telecommunications Department at CentraleSupélec, where he is currently a professor and holds the TCL Chair on 5G. He is also a researcher at the Laboratoire des Signaux et Systèmes (L2S, CNRS). He has co-authored 1 book and more than 100 publications in journals and conference proceedings and serves regularly as TPC member or TPC co-chair for top international conferences. He is an Editor for the IEEE Wireless Communications Letters and the Journal of Communications and Information Networks. He has given in the past successful tutorials on 5G systems at various conferences including IEEE ISWCS'15 and IEEE WCNC'16 conferences. His research interests include 5G networks, MIMO systems, mathematical modelling of communication networks, stochastic network optimization and Machine Learning in wireless networks. 
\end{IEEEbiographynophoto}


\begin{IEEEbiographynophoto}{Berna Sayrac}
 is a senior research expert in Orange Labs. She received the B.S., M.S. and Ph.D. degrees from the Department of Electrical and Electronics Engineering of Middle East Technical University (METU), Turkey, in 1990, 1992 and 1997, respectively. She worked as an Assistant Professor at METU between 2000 and 2001, and as a research scientist at Philips Research France between 2001 and 2002. Since 2002, she is working at Orange Labs. Her current research activities and interests include 5G radio access and spectrum. She has taken responsibilities and been active in several FP7 European projects, such as FARAMIR, UNIVERSELF, and SEMAFOUR. She was the technical coordinator of the Celtic+ European project SHARING on self-organized heterogeneous networks, and also the technical manager of the H2020 FANTASTIC-5G project. Currently, she is the coordinator of the research program on RAN design and critical communications in Orange. She holds several patents and has authored more than fifty peer-reviewed papers in prestigious journals and conferences. She also acts as expert evaluator for research projects, as well as reviewer, TPC member and guest editor for various conferences and journals.

\end{IEEEbiographynophoto}

\begin{IEEEbiographynophoto}{Azeddine Gati}
received the Engineer degree in telecommunication and signal processing in 1996, the M.Sc. degree from the University of Rennes, Rennes, France, and the Ph.D. degree from the University of Pierre and Marie Curie (ParisVI), Paris, France, in 2000. From 1997 to 2001, he was engaged in studying circuit optimization and electromagnetic computational methods. Since 2001, he has been at Orange Labs, France Telecom, Issy les Moulineaux, France, where he is involved in research in the fields of applied electromagnetic for telecommunications systems. He is also involved in research on human interaction with radio waves. Since 2007, he has been leading research projects on sustainable development, including interactions of waves with human bodies, body area networks, wireless network planning tools, and energy in information, communication technology and service solutions. Dr. Gati is a member of the Union Radio-Scientifique Internationale (URSI), France.
\end{IEEEbiographynophoto}

\end{document}